\documentclass[11pt]{article}

\usepackage[nottoc]{tocbibind}

\usepackage{amsfonts}
\usepackage{amssymb}
\usepackage{amstext}
\usepackage{amsmath}
\usepackage{mathtools}

\usepackage{amsthm} 

\usepackage{caption}
\usepackage{subcaption}

\usepackage{color}
\usepackage{nameref}
\definecolor{ForestGreen}{rgb}{0.1333,0.5451,0.1333}
\definecolor{DarkRed}{rgb}{0.8,0,0}
\definecolor{Red}{rgb}{0.9,0,0}
\usepackage[linktocpage=true,
	pagebackref=true,colorlinks,
	linkcolor=DarkRed,citecolor=ForestGreen,
	bookmarks,bookmarksopen,bookmarksnumbered]
	{hyperref}
\usepackage{cleveref}

\usepackage{thmtools,thm-restate} 

\usepackage[numbers,sort&compress]{natbib}

\usepackage{graphicx}
\usepackage{graphics}
\usepackage{colordvi}
\usepackage{xspace}
\usepackage{algorithm}
\usepackage{algorithmicx}
\usepackage{url}
\usepackage{enumitem}

\newenvironment{proofof}[1]{\noindent{\bf Proof of #1.}}%
        {\hspace*{\fill}$\Box$\par\vspace{4mm}}

\newcommand{\enc}[1]{\langle #1 \rangle}

\usepackage[left=1in,top=1in,right=1in,bottom=1in]{geometry} 
\usepackage{mathptmx} 
\usepackage[scaled=.90]{helvet}
\usepackage{courier}
\usepackage{booktabs}
\usepackage{threeparttable}
\usepackage[small,compact]{titlesec}
\setlist[itemize]{noitemsep,nolistsep} 
\setlist[enumerate]{noitemsep,nolistsep} 

\makeatletter
\renewcommand{\paragraph}{%
  \@startsection{paragraph}{4}%
  {\z@}{1ex \@plus 1ex \@minus .2ex}{-1em}%
  {\normalfont\normalsize\bfseries}%
}
\makeatother

\makeatletter
\def\thmt@refnamewithcomma #1#2#3,#4,#5\@nil{%
  \@xa\def\csname\thmt@envname #1utorefname\endcsname{#3}%
  \ifcsname #2refname\endcsname
    \csname #2refname\expandafter\endcsname\expandafter{\thmt@envname}{#3}{#4}%
  \fi
}
\makeatother

\declaretheorem[numberwithin=section,refname={Theorem,Theorems},Refname={Theorem,Theorems}]{theorem}
\declaretheorem[numberlike=theorem,refname={Lemma,Lemmas},Refname={Lemma,Lemmas}]{lemma}
\declaretheorem[numberlike=theorem,refname={Conjecture,Conjectures},Refname={Conjecture,Conjectures}]{conjecture}
\declaretheorem[numberlike=theorem,refname={Corollary,Corollaries},Refname={Corollary,Corollaries}]{corollary}
\declaretheorem[numberlike=theorem,refname={Proposition,Propositions},Refname={Proposition,Propositions}]{proposition}
\declaretheorem[numberlike=theorem,refname={property,properties},Refname={Property,Properties}]{property}
\declaretheorem[numberlike=theorem,refname={observation,observations},Refname={Observation,Observations}]{observation}
\declaretheorem[numberlike=theorem,refname={assumption,assumptions},Refname={Assumption,Assumptions}]{assumption}
\declaretheorem[numberlike=theorem,refname={Claim, Claims},Refname={Claim, Claims}]{claim}
\declaretheorem[numberlike=theorem,refname={hypothesis,hypothesis},Refname={Hypothesis,Hypothesis}]{hypothesis}
\declaretheorem[numberlike=theorem]{definition}
\declaretheorem[style=remark]{remark}

\newtheorem{invariant}[theorem]{Invariant}
\newtheorem{assume}{Assumption}
\newtheorem{fact}{Fact}
\newtheorem{Rule}[theorem]{Rule}
\newcommand{\size}[1]{\ensuremath{\left|#1\right|}}
\newcommand{\ceil}[1]{\ensuremath{\left\lceil#1\right\rceil}}
\newcommand{\floor}[1]{\ensuremath{\left\lfloor#1\right\rfloor}}
\newcommand{\Dexp}[1]{\dexp\{#1\}}
\newcommand{\Tower}[2]{\operatorname{tower}^{(#1)}\{#2\}}
\newcommand{\logi}[2]{\operatorname{log}^{(#1)}{#2}}
\newcommand{\norm}[1]{\lVert #1\rVert}
\newcommand{\abs}[1]{\lvert #1\rvert}
\newcommand{\paren}[1]{\left ( #1 \right ) }
\newcommand{\union}{\cup}
\newcommand{\band}{\wedge}
\newcommand{\bor}{\vee}
\newcommand{\prof}{\ensuremath{\operatorname{profit}}}
\newcommand{\pay}{\ensuremath{\operatorname{pay}}}

\newcommand{\scale}[2]{\scalebox{#1}{#2}}

\newcommand{\fig}[1]{
\begin{figure}[h]
\rotatebox{0}{\includegraphics{#1}}
\end{figure}}

\newcommand{\figcap}[2]
{
\begin{figure}[h]
\rotatebox{270}{\includegraphics{#1}} \caption{#2}
\end{figure}
}

\newcommand{\scalefig}[2]{
\begin{figure}[h]
\scalebox{#1}{\includegraphics{#2}}
\end{figure}}

\newcommand{\scalefigcap}[3]{
\begin{figure}[h]
\scalebox{#1}{\rotatebox{0}{\includegraphics{#2}}} \caption{#3}
\end{figure}}

\newcommand{\scalefigcaplabel}[4]{
\begin{figure}[h]
\begin{center}
\label{mainbody:#4} \scalebox{#1}{\includegraphics{#2}}\caption{#3}
\end{center}
\end{figure}}

\newenvironment{prog}[1]{
\begin{minipage}{5.8 in}
{\sc\bf #1}
\begin{enumerate}}
{
\end{enumerate}
\end{minipage}
}

\renewcommand{\phi}{\varphi}
\newcommand{\Sum}{\displaystyle\sum}
\newcommand{\half}{\ensuremath{\frac{1}{2}}}

\newcommand{\poly}{\operatorname{poly}}
\newcommand{\dist}{\mbox{\sf dist}}

\newcommand{\reals}{{\mathbb R}}

\newcommand{\expct}[1]{\text{\bf E}_\left [#1\right]}
\newcommand{\expect}[2]{\text{\bf E}_{#1}\left [#2\right]}
\newcommand{\prob}[1]{\text{\bf Pr}\left [#1\right]}
\newcommand{\pr}[2]{\text{\bf Pr}_{#1}\left [#2\right ]}

\newcommand{\notsat}[1]{\overline{\text{SAT}(#1)}}
\newcommand{\notsats}[2]{\overline{\text{SAT}_{#1}(#2)}}

\newcommand{\sndp}{\mbox{\sf SNDP}}
\newcommand{\ecsndp}{\mbox{\sf EC-SNDP}}
\newcommand{\vcsndp}{\mbox{\sf VC-SNDP}}
\newcommand{\kec}{k\mbox{\sf -edge connectivity}}
\newcommand{\kvc}{k\mbox{\sf -vertex connectivity}}
\newcommand{\sskec}{\mbox{\sf single-source}~k\mbox{\sf -edge connectivity}}
\newcommand{\sskvc}{\mbox{\sf single-source}~k\mbox{\sf -vertex connectivity}}
\newcommand{\subkvc}{\mbox{\sf subset}~k\mbox{\sf -vertex connectivity}}
\newcommand{\oneec}{1\mbox{\sf -edge connectivity}}
\newcommand{\onevc}{1\mbox{\sf -vertex connectivity}}
\newcommand{\kvcssp}{k\mbox{\sf -vertex-connected spanning subgraph problem}}
\newcommand{\BLM}{{\sf Bounded-Length Multicut}\xspace}
\newcommand{\bydef}{\stackrel{{\triangle}}{=}}
\newcommand{\eblue}{E^{\mbox{\scriptsize{blue}}}}
\newcommand{\ered}{E^{\mbox{\scriptsize{red}}}}

\newcommand{\induce}[1]{\ensuremath{{\sf im}(#1)}\xspace}
\newcommand{\sinduce}[1]{\ensuremath{{\sf sim}(#1)}\xspace}
\newcommand{\sinducesigma}[2]{\ensuremath{{\sf sim}_{#1}(#2)}\xspace}

\newcommand{\SMP}{{\sf SMP}\xspace}
\newcommand{\UDP}{{\sf UDP}\xspace}
\newcommand{\semiind}{{\sf Semi-Induced Matching}\xspace}
\newcommand{\LSMP}{{\sf Limit-SMP}\xspace}
\newcommand{\LUDP}{{\sf Limit-UDP}\xspace}
\newcommand{\ETH}{{\sf ETH}\xspace}
\newcommand{\SAT}{{\sf SAT}\xspace}
\newcommand{\p}{{\bf p}}
\newcommand{\e}{{\sc e}}
\newcommand{\val}{{\sf val}\xspace}

\newcommand{\wall}{\overrightarrow{R}}
\newcommand{\awall}{\widehat{R}}
\newcommand{\pairs}{\mathcal{S}}

\newcommand{\edp}{{\sf edp}\xspace}
\newcommand{\tildeEDP}{\widetilde{\sf edp}\xspace}
\newcommand{\vdp}{{\sf vdp}\xspace}
\newcommand{\mis}{{\sf mis}\xspace}
\newcommand{\NDP}{{\sf VDP}\xspace}
\newcommand{\VDP}{{\sf VDP}\xspace}
\newcommand{\EDP}{{\sf EDP}\xspace}

\newcommand{\NN}{\mathcal{N}}

\newcommand{\reminder}[1]{{\it\bf \color{red}**#1**}}

\newcommand{\QValue}{{\sc QueryValue}\xspace}
\newcommand{\QSet}{{\sc QuerySet}\xspace}

\renewcommand{\t}{\text{slack}}
\newcommand{\s}{\text{status}}
\newcommand{\N}{\mathcal{N}}
\newcommand{\B}{\mathcal{B}}
\newcommand{\C}{\mathcal{C}}
\renewcommand{\P}{\mathcal{P}}
\newcommand{\eps}{\delta}
\newcommand{\dl}{\delta}
\newcommand{\dd}{D}
\newcommand{\G}{\mathcal{G}}
\newcommand{\V}{\mathcal{V}}
\newcommand{\E}{\mathcal{E}}
\newcommand{\MM}{\mathcal{\M}}
\newcommand{\M}{\mathcal{M}}
\newcommand{\OO}{\mathcal{O}}
\newcommand{\D}{\mathcal{D}}

\ifdefined\ShowComment

\def\danupon#1{\marginpar{$\leftarrow$\fbox{D}}\footnote{$\Rightarrow$~{\sf #1 --Danupon}}}
\def\sayan#1{\marginpar{$\leftarrow$\fbox{S}}\footnote{$\Rightarrow$~{\sf #1 --Sayan}}}
\def\monika#1{\marginpar{$\leftarrow$\fbox{M}}\footnote{$\Rightarrow$~{\sf #1 --Monika}}}

\else

\def\danupon#1{}
\def\sayan#1{}
\def\monika#1{}

\fi

\newboolean{short}
\setboolean{short}{true} 

\newcommand{\shortOnly}[1]{\ifthenelse{\boolean{short}}{#1}{}}
\newcommand{\longOnly}[1]{\ifthenelse{\boolean{short}}{}{#1}}

\title{Fully Dynamic Approximate  Maximum Matching and Minimum Vertex Cover  in $O(\log^3 n)$ Worst Case Update Time\footnote{An extended abstract of this paper appeared in SODA 2017.}}

\date{}

\author{Sayan Bhattacharya\thanks{University of Warwick, UK. Email: {\tt s.bhattacharya@warwick.ac.uk}}  \and Monika Henzinger\thanks{University of Vienna, Austria. Email: {\tt monika.henzinger@univie.ac.at}. The research leading to these results has received funding from the European Research Council under the European Union's Seventh Framework Programme (FP/2007-2013) / ERC Grant Agreement no. 340506.}   \and Danupon Nanongkai\thanks{KTH Royal Institute of Technology, Sweden. Email: {\tt danupon@gmail.com}. Supported by Swedish Research Council grant 2015-04659  ``Algorithms and Complexity for Dynamic Graph Problems''.}}

\begin{document}

\setcounter{tocdepth}{3}

\begin{titlepage}
\maketitle
\pagenumbering{roman}

\begin{abstract}
We consider the problem of maintaining an approximately maximum  (fractional) matching and an approximately minimum  vertex cover in a dynamic graph. Starting with the seminal paper by Onak and Rubinfeld [STOC 2010], this problem has received significant attention in recent years. There remains, however, a polynomial gap between the best known worst case update time and the best known amortised update time for this problem, even after allowing for randomisation. Specifically, Bernstein and Stein [ICALP 2015, SODA 2016] have the best known worst case update time. They present a deterministic data structure with approximation ratio $(3/2+\epsilon)$ and worst case update time $O(m^{1/4}/\epsilon^2)$, where $m$ is the number of edges in the graph. In recent past, Gupta and Peng [FOCS 2013] gave a deterministic data structure with approximation ratio $(1+\epsilon)$ and worst case update time $O(\sqrt{m}/\epsilon^2)$. No known randomised data structure  beats the worst case update times of these two results. In contrast, the paper by Onak and Rubinfeld [STOC 2010]  gave a randomised data structure with approximation ratio $O(1)$ and amortised update time $O(\log^2 n)$, where $n$ is the number of nodes in the graph. This  was later improved by Baswana, Gupta and Sen [FOCS 2011] and Solomon [FOCS 2016], leading to a randomised date structure with approximation ratio $2$ and amortised update time $O(1)$.

We bridge the polynomial gap between the worst case and amortised update times for this problem, without using any randomisation. We present a deterministic data structure with approximation ratio $(2+\epsilon)$ and worst case update time $O(\log^3 n)$, for all sufficiently small constants $\epsilon$.
\end{abstract}

\newpage
\pagenumbering{gobble}
\clearpage

\setcounter{tocdepth}{3}


\end{titlepage}

\newpage

\pagenumbering{arabic}

\newcommand{\Up}{\text{{\sc Up}}}
\newcommand{\Down}{\text{{\sc Down}}}
\newcommand{\Idle}{\text{{\sc Idle}}}
\newcommand{\Upb}{\text{{\sc Up-B}}}
\newcommand{\Downb}{\text{{\sc Down-B}}}
\newcommand{\state}{\text{{\sc State}}}
\newcommand{\Slack}{\text{{\sc Slack}}}
\newcommand{\polylog}{\text{polylog}}

\section{Introduction}
\label{sec:intro}

A matching in a graph is a set of edges that do not share any common endpoint. In the dynamic matching problem, we want to maintain an (approximately) maximum-cardinality matching when the  input graph is undergoing edge insertions and deletions. 
The time taken  to handle an edge insertion or deletion in the input graph is called the {\em update time} of the concerned dynamic algorithm. Our goal is this paper is to design a dynamic algorithm whose update time is as small as possible. Throughout this paper, we denote the number of nodes and edges in the input graph  by $n$ and $m$ respectively. The value of $n$ remains fixed over time, since the set of nodes in  the graph remains the same. However, the value of $m$ changes   as  edges get inserted or deleted in the graph. 
Similar to static  problems where we want the running time of an algorithm to be polynomial in the input size, in the dynamic setting we desire the update time to be $\polylog(n)$, for an input  (edge insertion or deletion) to a dynamic problem can be specified using  $O(\log n)$ bits. 

The dynamic matching problem has been extensively studied in the past few years. We now know that within $\polylog(n)$ update time we can maintain a $2$-approximate matching using a randomized algorithm \cite{Solomon16,BaswanaGS11,OnakR10} and a $(2+\epsilon)$-approximate matching using a deterministic algorithm \cite{BhattacharyaHN16,BhattacharyaHI15s,BhattacharyaHI15}. The downside of these algorithms, however,  is that their update times are {\em amortised}. Thus, the algorithms take $\polylog(n)$ update time {\em on average}, but from time to time they may take as large as $O(n)$ time to respond to a single update. It is much more desirable to be able to guarantee a small update time after every update. This type of update time is called {\em worst-case update time}.

Unfortunately, known worst-case update time bounds for this problem are  {\em polynomial in $n$}:
the  known algorithms take $O(n^{1.495})$ worst-case update time to maintain the value of the maximum matching exactly \cite{Sankowski07} (also see \cite{AbboudW14,HenzingerKNS15-oMv,KopelowitzPP16} for complementing lower bounds), $O(\sqrt{m}/\epsilon^2)$ time to maintain a $(1+\epsilon)$-approximate maximum matching~\cite{GuptaP13,NeimanS13}, $O(m^{1/3}/\epsilon^2)$ time to maintain a $(4+\epsilon)$-approximate maximum matching \cite{BhattacharyaHI15s}, and $O(m^{1/4}/\epsilon^2)$ time to maintain a $(3/2+\epsilon)$-approximate maximum matching in bipartite graphs \cite{BernsteinS15,BernsteinS16}. There is no algorithm with $\polylog(n)$ worst-case update time even with a $\polylog(n)$ approximation ratio.

We note that the lack of a data structure with good worst-case update time  is not at all specific to the problem of dynamic matching. Other fundamental dynamic graph problems, such as spanning tree, minimum spanning tree and shortest paths also suffer the same issue~(see, e.g., \cite{Frederickson85,Frederickson97,HenzingerK99,HolmLT01,Wulff-Nilsen13a,DemetrescuI03,Thorup05}). One exception is the celebrated randomized algorithm with $\polylog(n)$ update time for dynamic connectivity~\cite{KapronKM13}. To the best of our knowledge,  our result is the first deterministic fully-dynamic graph algorithm with  $\polylog(n)$ worst-case update time in {\em general graphs}. In contrast, for a special class of graphs with arboricity bounded by $\alpha$ (say), the papers~\cite{orient1,orient2} present deterministic dynamic algorithms with $\tilde{O}(\alpha)$ worst case update times for the problem of maintaining edge orientation.

\paragraph{Our result.}
We present a deterministic algorithm that  maintains a {\em fractional matching}\footnote{In a fractional matching each edge is assigned a nonzero weight, ensuring that for every node the sum of the weights of the edges incident to it is at most $1$. The size of a fractional matching is the sum of the weights of all the edges in the graph.} and  a {\em vertex cover}\footnote{A vertex cover is a set of nodes such that every edge in the graph has at least one endpoint in that set.} whose sizes are within a $(2+\epsilon)$ factor of each other, for all sufficiently small constants  $\epsilon$. Since the size of a maximum fractional matching is at most $3/2$ times the size of a maximum matching,  we can also maintain a $(3+\epsilon)$-approximation to the {\em size} of the maximum matching in $O(\log^3 n)$ worst-case update time.

\section{A high level overview of our algorithm}
\label{sec:overview}
In this section, we  present the main  ideas behind our algorithm. The formal description of the algorithm and the analysis   appears in subsequent sections.

\paragraph{Hierarchical Partition.} Our algorithm builds on the ideas from a dynamic data structure of Bhattacharya, Henzinger and Italiano~\cite{BhattacharyaHI15s} called {\em $(\alpha, \beta)$-decomposition}. This data structure maintains a $(2+\epsilon)$-approximate maximum fractional matching in $O(\log n/\epsilon^2)$ amortised update time. It defines the fractional edge weights using {\em levels} of nodes and edges. In particular, fix two constants $\alpha, \beta \geq 1$, and recall that the input graph $G = (V, E)$ has $|V| = n$ nodes. Partition the node set $V$ into $L +1$ levels $\{0, \ldots, L\}$, where $L = \log_{\beta} n$. Let $\ell(y) \in \{0, \ldots, L\}$ denote the level of a node $y \in V$. 
The level of an edge $(x, y)$ is given by Eq.~\eqref{eq:level_old}, and we assign a fractional weight $w(x, y)$ as  per Eq.~\eqref{eq:weight:overview}.\
\begin{eqnarray}
\ell(x,y) & = & \max(\ell(x), \ell(y)) \label{eq:level_old} \\
w(x, y) & = & \beta^{-\ell(x, y)}\label{eq:weight:overview}
\end{eqnarray}
Thus, the weight of an edge decreases exponentially with its level. The weight of a node $y \in V$ is defined as $W_y = \sum_{(x,y) \in E} w(x,y)$. This equals the sum of the weights of the edges incident on it. The goal is to maintain a partition satisfying the following property.

\begin{property}
\label{inv:prev algo}
Every node $y$ with $\ell(y)>0$ has weight $1/(\alpha\beta) \leq W_y < 1$. Furthermore, every node $y$ with $\ell(y)=0$ has weight $0 \leq W_y < 1$. 
\end{property}

To provide some intuition, we  show how to construct a hierarchical partition satisfying Property~\ref{inv:prev algo} in the  static setting, when there is no edge insertions/deletions. For notational convenience, we define $V^*_L = V$. Initially, we put all the nodes in level $L$,  and as per equations~\ref{eq:level_old},~\ref{eq:weight:overview} we assign a weight $w(x, y) = \beta^{-L} = 1/n$ to every edge $(x, y) \in E$. Since every node has degree at most $n-1$, we get $0 \leq W_y < 1$ for all $y \in V^*_L$. We now execute a {\sc For} loop as follows.

\begin{itemize}
\item
{\sc For} $i = L$ to $1$:
\begin{itemize}
\item We partition the node-set $V^*_i$ into two subsets: $V_i = \{ y \in V : 1/\beta \leq W_y < 1\}$ and $V^*_{i-1} = \{ y \in V : 0 \leq W_y < 1/\beta \}$.  Next, we move down the nodes in $V^*_{i-1}$ to level $i-1$. The level and weight of every edge incident on a node in $V \setminus V^*_{i-1} = V_i \cup \ldots \cup V_L$ remain unchanged during this step, as per equations~\ref{eq:level_old} and~\ref{eq:weight:overview}. Hence, just after the nodes in $V^*_{i-1}$ are moved down to level $i-1$, we get $1/\beta \leq W_y < 1$ for all nodes $y$ at level $i$. The weights of the remaining  edges (whose both endpoints lie in $V^*_{i-1}$) increase by a factor of $\beta$. Hence, the weights of the nodes in $V^*_{i-1}$ also increase by at most a factor of $\beta$.  Before the nodes in $V^*_{i-1}$ were moved down to level $i-1$, we had $0 \leq W_y < 1/\beta$ for all $y \in V^*_{i-1}$. Thus, just after the nodes in $V^*_{i-1}$ are moved down to level $i-1$, we get $0 \leq W_y < 1$ for all $y \in V^*_{i-1}$. 
\end{itemize}
\end{itemize}
\smallskip
\noindent When the above {\sc For} loop terminates, we  have $1/\beta \leq W_y < 1$ for all nodes $y \in V$ at levels $\ell(y) > 0$, and $0 \leq W_y < 1$ for all nodes $y \in V$ at level $\ell(y) = 0$. Specifically, Property~\ref{inv:prev algo} is satisfied with $\alpha = 1$.

\begin{theorem}[\cite{BhattacharyaHI15s}]
	\label{th:old:result}
	Under Property~\ref{inv:prev algo}, the edge-weights $\{ w(e) \}$ form a $2\alpha\beta$-approximate maximum fractional matching in $G$.
\end{theorem}

In \cite{BhattacharyaHI15s}, Bhattacharya~et~al. showed that we can dynamically maintain such a partition with $\alpha=\beta=(1+\epsilon)$ in $O(\log n/\epsilon^2)$ amortised update time. The main idea is as follows. Assume that we have a partition that satisfies Property~\ref{inv:prev algo}. Now an edge $(u, v)$ is inserted or deleted. 
This causes  $W_u$ and $W_v$ to increase or decrease. Hence, it might happen that some node $x \in \{u, v\}$ violates Property~\ref{inv:prev algo} after the insertion/deletion of the edge $(u, v)$, i.e. either (1) $W_x \geq 1$ or (2) $W_x < 1/(\alpha \beta)$ and $\ell(x) > 0$. We call such a node $x$ {\em dirty}, and deal with this event  by changing the level of $x$ in a  straightforward way as per Figure~\ref{alg:fix dirty old}: If $W_x$ is too large (resp. too small), then we increase (resp. decrease) $\ell(x)$ by one. This causes the weights of some edges incident on $x$ to decrease (resp. increase), which in turn  decreases (resp. increases) the value of $W_x$. For each level $i \in [0, L]$, we define the set of edges $E_i(x)$ as follows.
\begin{align}
E_{i}(x)=\{(x, y)\in E \mid \ell(x, y)=i\}. \label{eq:E_i:overview}
\end{align}
An important observation is that as a node $x$ moves up (resp. down) from level $i$ to level $i+1$ (resp. $i-1$), the edges whose weights get changed all belong to the set $E_i(x)$. 
Since  the relevant data structures can be maintained efficiently, this implies that the runtime of one iteration of the {\sc While} loop in Figure~\ref{alg:fix dirty old} is dominated by the cost of Line~7, which takes $O(|E_i(x)|)$ time. In~\cite{BhattacharyaHI15s}, the authors showed that this cost can be amortised over previous edge insertions/deletions. 

Note that   one iteration of the {\sc While} loop can make some neighbours of $x$  dirty, and $x$ itself might remain dirty at the end of the iteration.
These dirty nodes are dealt with in subsequent iterations in a similar way (until there is no dirty node left).

\begin{figure}[htbp]
\centerline{\framebox{
\begin{minipage}{5.5in}
\begin{tabbing}
01.  \=  {\sc While} there is a dirty node $x$ \\
02. \> \qquad \=  Let $i=\ell(x)$\\
03. \> \> {\sc If} $W_x \geq 1$, {\sc Then} \ \  // In this case $i < L$ \\ 
04. \> \> \qquad \= Set $\ell(x) \leftarrow \ell(x) + 1$. \\ 
05. \> \> {\sc Else} \ \ // In this case $W_x < 1/(\alpha\beta)$, $i > 0$ \\
06. \> \> \> Set $\ell(x) \leftarrow \ell(x) - 1$. \\
07.	\> \> Update $\ell(x, y)$ for all $(x, y)\in E_i(x)$. 
\end{tabbing}
\end{minipage}
}}
\caption{\label{alg:fix dirty old} Fixing the dirty nodes.}
\end{figure}

\paragraph{Example: Inserting edges to a star.}
The following example shows the basic idea behind the amortisation argument.
Consider a star centred at node $v$ consisting of $\beta^{i-1}$ edges, for some large $i$. To satisfy Property~\ref{inv:prev algo}, we can set $\ell(v)=i$, while all other nodes have level $0$. Thus, we get $W_v=1/\beta$ since every edge has weight $1/\beta^i$. Now keep inserting edges to the star (the graph remains a star throughout). Property~\ref{inv:prev algo} remains satisfied until the $(\beta^{i} - \beta^{i-1})$-th edge is inserted -- at this point the star consists of $\beta^i$ edges, $W_v=1$, and the node $v$ becomes dirty. We fix the node by increasing $\ell(v)$ to $i+1$ as in Algorithm~\ref{alg:fix dirty old}, thus reducing the edge-weights to $1/\beta^{i+1}$ and the value of $W_v$ to $1/\beta$. To do this we have to pay the cost of $O(|E_{i}(v)|)= O(\beta^i)$ in terms of update time. We can amortise this cost over the $(\beta^i - \beta^{i-1})$ newly inserted edges. This gives an amortised update time of $O(1)$ for constant $\beta$.

Note that in the above example the algorithm does not perform well in the worst case:
after the  $(\beta^{i}-\beta^{i-1})$-th insertion it has to ``probe'' all edges in $E_{i}(v)$. So the worst case update time becomes $O(\beta^i)$, which can be polynomial in $n$ when $i$ is large. 
But in this particular instance the problem can be fixed easily: Whenever $v$ becomes dirty due to the insertion of an edge with weight $1/\beta^i$, we reduce  the weight of the newly inserted edge and some other edge in the star from $1/\beta^i$ to  $1/\beta^{i+1}$. Thus, the net increase in the weight of $v$ becomes equal to $1/\beta^i - 2(1/\beta^i - 1/\beta^{i+1}) = 2/\beta^{i+1} - 1/\beta^i \leq 0$ (the last inequality holds as long as $\beta \geq 2$). In other words, when the node $v$ becomes dirty, by reducing the weights of two edges to $1/\beta^{i+1}$ we can ensure that $W_v$ again becomes smaller than one. Once every edge has weight $1/\beta^{i+1}$, we set $\ell(v)=i+1$.

\paragraph{Shadow-level ($\ell_y(x, y)$).} 
To make the above idea concrete, we introduce the notion of a {\em shadow-level}. For every node $y\in V$ and every incident edge $(x, y) \in E$, we define the {\em shadow-level of $y$ with respect to $(x, y)$}, denoted by $\ell_y(x, y) \in \{0, \ldots , L\}$, to be an integer in $\{0, \ldots, L\}$ such that the following property holds.
\begin{property}
	\label{inv:shadow:level:overview}
	For every node $y$ and edge $(x, y)$, $\ell(y) - 1 \leq \ell_y(x, y) \leq \ell(y) + 1$. 
\end{property}

We modify the definition of the level of an edge $(x, y) \in E$ (in Eq.~\eqref{eq:level_old}) to
\begin{align}
\ell(x, y) = \max(\ell_x(x, y), \ell_y(x, y)).\label{eq:level_new}
\end{align}
This affects the value of $w(x, y)$ and the set $E_i(y)$ as they depend on the levels of edges (see Eq.~\eqref{eq:weight:overview} and \eqref{eq:E_i:overview}).
The idea of the shadow-level is that if $\ell_y(x, y)>\ell(y)$ (respectively $\ell_y(x, y)<\ell(y)$), then from the perspective of the edge $(x, y)$ we have already increased (resp. decreased) $\ell(y)$; thus, the level and weight of $(x, y)$ has changed accordingly.
In this case, we say that $y$ {\em up-marks} (respectively {\em down-marks}) the edge $(x, y)$. We will use this operation when $W_y$ is too large (resp. too small). 
Intuitively, $y$ should not up-mark and down-mark edges at the same time. In particular, let $M_{down}(y) = \{ (x, y) \in E : \ell_y(x, y) = \ell(y) - 1\}$ and $M_{up}(y) = \{ (x, y) \in E : \ell_y(x, y) = \ell(y) + 1 \}$ respectively denote the set of all edges down-marked and up-marked by $y$. Then, we will maintain the following property. 
\begin{property}\label{inv:up or down only}
	Either $M_{up}(y)=\emptyset$ or $M_{down}(y)=\emptyset$. 
\end{property}

To see the usefulness of this new definition, consider the following algorithm for dealing with 
the case where the graph is always a star centred at $v$:
If there are only edge insertions, then $v$ up-marks the newly inserted edge and another edge in $E_{\ell(v)}(v)$ whenever it  becomes dirty (i.e. $W_v \geq 1$). It is easy to see that this will be enough to keep $W_v<1$ as long as $\beta \geq 2$. Once $E_{\ell(v)}(v)=\emptyset$, we increase $\ell(v)$ by one.
Similarly, if there are only edge deletions, then $v$ can down-mark an edge in $E_{\ell(v)}(v)$ whenever it becomes dirty (i.e. $W_v<1/(\alpha\beta)$).

The algorithm follows the same strategy when there are both edge insertions and deletions, albeit with one caveat: To ensure that Property~\ref{inv:up or down only} holds, it cannot up-mark an edge  if $M_{down}(v)\neq \emptyset$, and cannot down-mark an edge  if $M_{up}(v)\neq \emptyset$. Suppose that $M_{down}(v) \neq \emptyset$ and we want to reduce the weight of $v$. In this event, the node $v$ picks an edge $(u, v)$ in $M_{down}(v)$ and sets $\ell_v(u, v)$ back from $\ell(v)-1$ to $\ell(v)$, which reduces the value of $w(u, v)$. This  causes the edge $(u, v)$ to be removed from $M_{down}(v)$ and be added to $E_{\ell(v)}(v)$. We say that the node $v$ {\em un-marks} the edge $(u, v)$. Next, suppose that $M_{up}(v) \neq \emptyset$ and we want to increase the weight of $v$.
In this event, the node $v$ un-marks an edge in  $M_{up}(v)$.

\begin{figure}
	\centering
		\includegraphics[width= .6\linewidth]{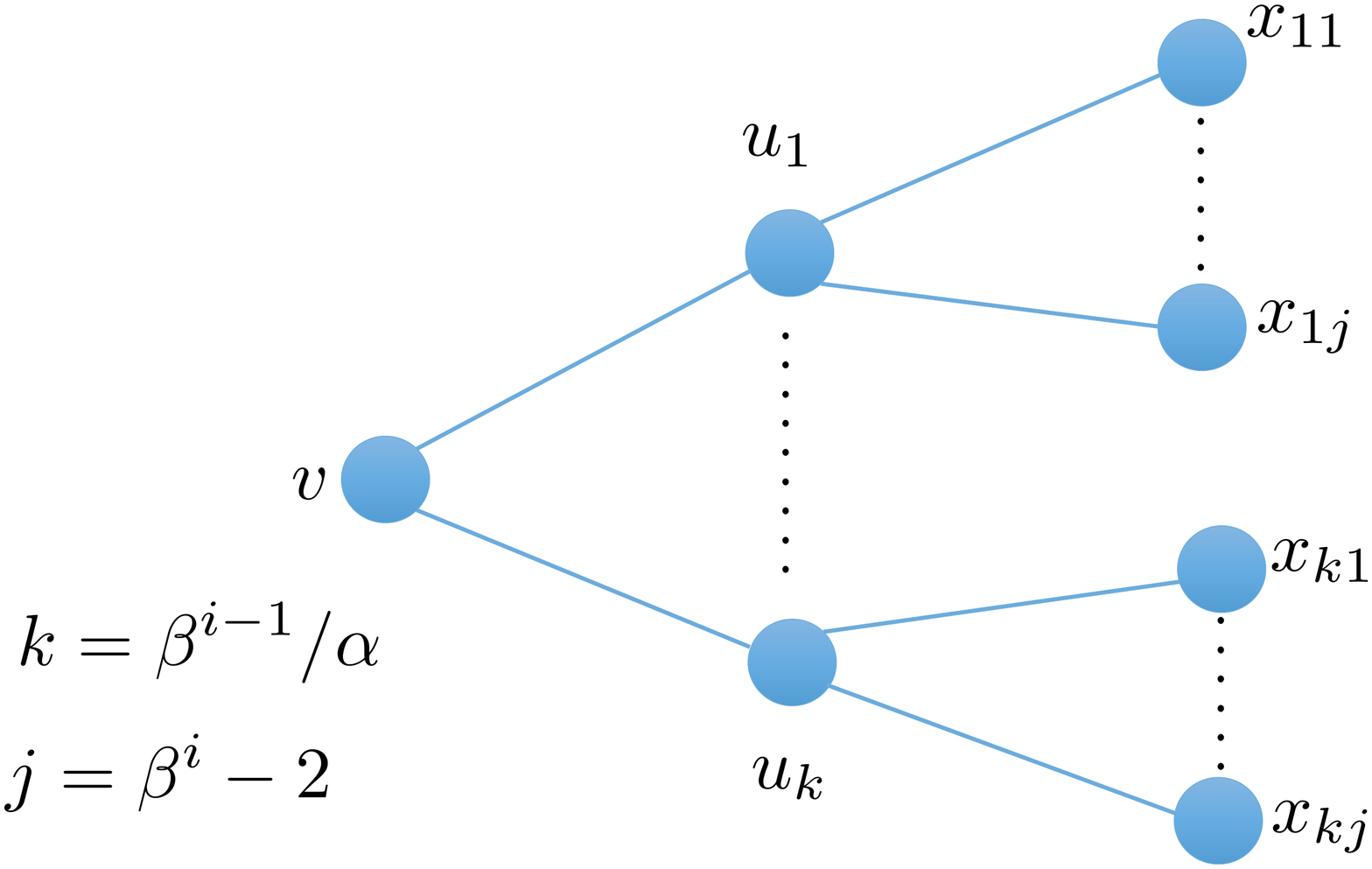}
	\caption{Example. Suppose that $i$ is very large, e.g. $i=(\log n)/2$, and  $\beta$ is a large constant.}\label{fig:examples}
\end{figure}

\paragraph{Failures.} So far we have described an idea that leads to small worst-case update time when the input instance is a star graph. To make  this idea work  on a general input instance, we have to deal with  several issues that make the algorithm more complicated. 
The chief one among them is the observation that the algorithm may {\em fail} in adjusting edge weights. 
Consider, for example, a tree rooted at a node $v$ having $k=\beta^{i-1}/\alpha$ children, say $u_1, \ldots u_{k}$. Further, each $u_i$ has $j = \beta^{i}-2$ children. See Fig.~\ref{fig:examples}.
Suppose that we satisfy Property~\ref{inv:prev algo} by setting $\ell(x_{p, q}) \leftarrow 0$ for every leaf-node $x_{p, q}$ with $p \in [k]$ and $q \in [j]$,  $\ell(u_p) \leftarrow i$ for every internal node $u_p$ with $p \in [k]$, and $\ell(v) \leftarrow i$ for the root node $v$. No edge is down-marked or up-marked by any node. This implies that $W_v=1/(\alpha\beta)$, $W_{u_p}=1-1/\beta^i$ for all $p \in [k]$, and $W_{x_{p, q}} = 1/\beta^i$ for all $p \in [k]$ and $q \in [j]$. 

Now, suppose that the edge $(v, u_1)$ gets deleted. This makes   $v$ dirty, for  $W_v$ becomes smaller than $1/(\alpha\beta)$. The algorithm  responds by down-marking an edge in $E_i(v)$, say $(v, u_2)$. Unfortunately, this down-marking does not change the weight  $w(u, v_2)$ since $\ell_{v_2}(u, v_2)=i$. I We say that this down-marking {\em fails}. 
When a down-marking fails, the node $v$ remains dirty and Property~\ref{inv:prev algo} remains unsatisfied. In fact, in this example, the node $v$ will remain dirty even if we down-mark all the edges in $E_i(v)$. To satisfy Property~\ref{inv:prev algo}, we have no other option but to set $\ell(v) \leftarrow 0$. However, we cannot do so unless we probe all the edges in $E_i(v)$, for we have to ensure that all the down-markings on these edges fail. This takes too much time. 

To deal with this issue, we  keep  down-marking the edges in $E_i(v)$ as long as we fail, until the point when  we experience  $\polylog(n)$ failures. We might still end up  having $W_v<1/(\alpha\beta)$. Nevertheless,  we will  guarantee a constant approximation ratio by arguing that we continue to have $W_v = \Omega(1/(\alpha\beta))$. Intuitively, every time $W_v$ decreases by $1/\beta^i$ because of these failures, we  down-mark many edges in the set $E_i(v)$. Since  $|E_i(v)| \leq \beta^{i-1}/\alpha$, we are able to down-mark all the edges in $E_i(v)$ before the value of $W_v$ becomes too small. At that point, we are ready to decrease the level of $v$.

Specifically, suppose that the edges incident to  $v$ keep getting deleted. While handling $t=\beta^{i-1}/(\alpha\log^2 n)$ such deletions, we perform $t \cdot \polylog(n)\geq \beta^{i-1}/\alpha$ {\em failed} down-markings. This is enough to down-mark every edge in $E_i(v)$. At this point we move $v$ down to level $(i-1)$, and we still have $W_v\geq 1/(\alpha\beta)-t/\beta^i=(1-1/\log^2 n)/(\alpha\beta)$. By repeating this argument, we  conclude that even if there are more deletions, we still have $W_v= \Omega(1/(\alpha\beta))$ until the point in time when  $v$ moves down to level $0$  (where $v$ can not be dirty  for its weight being small).

\paragraph{The algorithm in a nutshell.} Our algorithm obeys the following principles. When a node $y$ becomes dirty, it either (1) up-marks or down-marks an edge $(x, y)$, or (2) un-marks an  edge $(x, y)$ if up-marking or down-marking would violate Property~\ref{inv:up or down only}.  Such an action may fail, meaning that $w(x,y)$ might not change, for  reasons  exemplified in Fig.~\ref{fig:examples}. In this event, $y$ continues probing its other incident edges, and stops when  it experiences either its first {\em success} or its $\polylog  (n)^{th}$ failure. We  show: (a) the failures do not cause the  weight $W_y$ to become too small or too large, (b) fixing one dirty node  leads to at most one new dirty node, and (c) the level of the dirty node under consideration drops after constantly many fixes. Item (a) guarantees a constant approximation factor. Items (b), (c) guarantee a $\polylog(n)$ update time, for there are $O(\log n)$ levels.

\section{Preliminaries}
\label{sec:algo}

Henceforth, we  focus on formally describing our dynamic algorithm and analysing its worst-case update time. For the rest of the paper, we fix two constants $\beta, K$ and  define $L$ and $f(\beta)$  as in equation~\ref{eq:beta:K:L}. Note that $K < L$ when $n$ is  sufficiently large.
\begin{equation}
\label{eq:beta:K:L}
\beta \geq 5, K = 20, f(\beta) = 1 - 3/\beta,  L = \lceil \log_{\beta} n \rceil.
\end{equation}

We will maintain a  hierarchical partition of the node-set in  $G = (V, E)$, and a fractional matching where the weights assigned to the edges depend on the levels of their endpoints. For technical reasons, however, there will be two key differences between the hierarchical partition actually used by our dynamic algorithm and the one that was defined in Section~\ref{sec:overview}. 

\begin{enumerate}
\item We will  collapse all the nodes  in levels $\{0, \ldots, K\}$ into a single level $K$. Accordingly, the level of a node will lie in the range $[K, L]$ in this new hierarchical partition. The weights of the nodes $y$ in levels $\ell(y) > K$ will satisfy the constraint: $f(\beta) \leq W_y < 1$. On the other hand, the weights of the nodes $y$ at the lowest level $\ell(y) = K$ will satisfy the constraint: $0 \leq W_y < 1$. Comparing these constraints with Property~\ref{inv:prev algo}, 	it follows that the term $1/(\alpha \beta)$ is replaced by $f(\beta)$ in the new partition. 
\item We will allow the weight of an edge $(u, v)$ to be off by a factor of $\beta$ from its ideal value $\beta^{-\max(\ell(u), \ell(v))}$. 
\end{enumerate}
\smallskip
\noindent The structure maintained by our  algorithm will be called a {\em nice-partition}. This is formally defined below.
\begin{definition}
\label{def:structure}
In a {\em nice-partition},  the node-set $V$ is partitioned into $(L-K+1)$ subsets $V_K, \ldots, V_L$. For $i \in [K, L]$, if a node $v$ belongs to  $V_i$, then we say that the node $v$ is at {\em level} $\ell(v) = i$. Each edge $(u, v) \in E$ gets a weight $w(u,v)$. Let $W_v = \sum_{(u,v) \in E} w(u,v)$ be the total weight received by a node $v$ from  its incident edges. The following  properties hold.
\begin{enumerate} 
\item For every edge $(u, v) \in E$, we have $\beta^{-\max(\ell(u), \ell(v)) -1} \leq w(u, v) \leq \beta^{-\max(\ell(u),\ell(v))+1}$. 
\item If a node  $v$ has $\ell(v) > K$, then $f(\beta) \leq W_v < 1$. 
\item If a node $v$ has $\ell(v) = K$, then $W_v < 1$. 
\end{enumerate}
\end{definition}
\begin{lemma}
\label{lm:main}
Suppose that we can maintain a nice-partition  in $O(T(n))$ worst-case update time. Then we can also maintain a $2/f(\beta)$-approximate maximum fractional matching and a $2/f(\beta)$-approximate minimum vertex cover   in $O(T(n))$ worst case update time.
\end{lemma}

\begin{proof}
Let  $E^r  = \{ (u,v ) \in E : \ell(u) = \ell(v) = K\}$ be the subset of edges with both endpoints at level $K$. We will maintain a {\em residual} weight $w^r(e) \geq 0$ for every edge $e \in E^r$. For notational consistency, we define $w^r(e) = 0$ for every edge $e \in E \setminus E^r$. Let $W^r_v = \sum_{(u, v) \in E^r} w^r(u,v)$ denote the residual weight received by a node $v$ from all its incident edges. Two conditions are satisfied: 
\begin{itemize}
\item (a) For each node $v \in V$, we have $0 \leq W_v + W^r_v \leq 1$.
\item (b) For every edge  $(u, v) \in E^r$, we have either $W_v + W^r_v \geq 1-1/\beta$  or  $W_u + W_u^r \geq 1-1/\beta$.
\end{itemize}
\smallskip
\noindent 
Let $\text{deg}^r(v)$ denote the degree of a node $v \in V$ among the edges in $E^r$. By condition~(1) of Definition~\ref{def:structure}, every edge $(x, y) \in E^r$ has weight $w(x, y) \geq \beta^{-K-1}$. Hence, for every node $v \in V$, we get: $1 > W_v \geq \sum_{(u, v) \in E^r} w(u, v) \geq \text{deg}^r(v) \cdot \beta^{-K-1}$.  This implies that $\text{deg}^r(v) < \beta^{K+1}$ for every node $v \in V$. Since $\beta, K$ are constants, we get: $\text{deg}^r(v) = O(1)$ for every node $v \in V$. 

\medskip
\noindent {\em Maintaining the residual weights $\{w^r(e)\}, e \in E^r$.}

\noindent For every node $v \in V$, let $b(v) =  1 - W_v$ denote the {\em capacity} of the node. Let $b^r(v)$ be  equal to the value of $b(v)$ rounded down to the nearest multiple of $1/\beta$. We say that $b^r(v)$ is the {\em residual capacity} of node $v$. We  create an {\em auxiliary  graph} $G^* = (V^*, E^*)$, where we have $\beta$ {\em copies} of each node $v \in V$. For every edge $(u, v) \in E^r$, there are $\beta^2$ edges in $G^*$: one for each pair of copies of $u$ and $v$. For each node $v \in V$, if $b^r(v) = t/\beta$ for some integer $t \in [0, \beta]$, then $t$ copies of $v$ are {\em turned on} in $G^*$,  and the remaining $(\beta - t)$ copies of $v$ are {\em turned off} in $G^*$. We maintain a maximal matching $M^*$ in the subgraph of $G^*$ induced by the copies of nodes that are turned on. Since $\text{deg}^r(v) = O(1)$ for every node $v \in V$, we can maintain the matching $M^*$ in $O(1)$ update time using a trivial algorithm. From the matching $M^*$, we get back the residual weights $\{w^r(e)\}$ as follows. For every edge $(u, v) \in E^r$, if there are $t$ edges in $M^*$ between different copies of $u$ and $v$, then we set $w^r(u,v) \leftarrow t/\beta$. It is easy to check that this  satisfies both  conditions (a) and (b).

\medskip
\noindent {\em Approximation guarantee.}

\noindent Condition (a) implies that the edge-weights $\{ w(e) + w^r(e) \}$ form a valid fractional matching in $G$. Define the subset of nodes $V^* = \{ v \in V : W_v + W^r_v \geq f(\beta) \}$. Consider any edge $(u, v) \in E$. If at least one endpoint $x \in \{u, v\}$ lies at a level $\ell(x) > K$, then condition (2) of Definition~\ref{def:structure} implies that $W_x + W^r_x \geq W_x \geq f(\beta)$, and hence $x \in V^*$. On the other hand, if both the endpoints $\{u, v\}$ lie at level $K$, then by conditions (a) and (b) we have:  $W_x + W^r_x \geq 1 - 1/\beta \geq f(\beta)$ for some $x \in \{u, v\}$, and hence $x \in V^*$. It follows that $V^*$ forms a valid vertex cover in $G$.   Applying  complementary slackness conditions, we infer that the edge-weights $\{ w(e) + w^r(e) \}$ form a $2/f(\beta)$-approximate maximum fractional matching in $G$, and that $V^*$ forms a $2/f(\beta)$-approximate minimum vertex cover in $G$.  
\end{proof}

Fix any constant $0 < \epsilon < 1$ and let $\beta = 3(2+\epsilon)/\epsilon$. Then  $\beta \geq 5$ and $2/f(\beta) = 2 +\epsilon$ (see equation~\ref{eq:beta:K:L}). Setting  $\beta$ in this way, we can use Theorem~\ref{th:main} and Lemma~\ref{lm:main} to maintain a $(2+\epsilon)$-approximate maximum fractional matching and a $(2+\epsilon)$-approximate minimum vertex cover in $O(\log^3 n)$ worst-case update time. We devote the rest of the paper to proving Theorem~\ref{th:main}.

\begin{theorem}
\label{th:main}
We can maintain a nice-partition  in $G = (V, E)$ in $O(\log^3 n)$ worst case update time. 
\end{theorem}

\subsection{Shadow-levels.}
\label{sub:sec:shadow-level}

As  in Section~\ref{sec:overview}, the shadow-levels will uniquely determine the weight $w(u,v)$ assigned to every edge $(u,v) \in E$. They will ensure that  $w(u,v)$ differs from the ideal value $\beta^{-\max(\ell(u), \ell(v))}$ by at most a factor of $\beta$. This implies condition (1) of Definition~\ref{def:structure}.
Specifically, we require that each edge  has two shadow-levels: one for each of its endpoints. Let $\ell_y(x, y) \in [K, L]$ be the shadow-level of a node $y$ with respect to the edge $(x, y)$. We require that this shadow-level can differ from the actual level of the node by at most one. This is formally stated in the invariant below.

\begin{invariant}
\label{inv:shadow:level}
For every node $y \in V$ and every edge $(x, y) \in E$, we have $\ell(y) - 1 \leq \ell_y(x, y) \leq \ell(y) + 1$.
\end{invariant}

Next, as in Section~\ref{sec:overview}, we define the {\em level of an edge} to be the maximum value among the shadow-levels of its endpoints. Let $\ell(x, y) \in [K,  L]$ be the level of an edge $(x, y)$. Then for every edge $(x, y) \in E$ we  have:
\begin{equation}
\label{eq:level:edge} 
\ell(x, y) = \max(\ell_x(x, y), \ell_y(x, y)).
\end{equation}

As in Section~\ref{sec:overview}, we now require that the weight assigned to an edge $(u, v) \in E$ be given by $\beta^{-\ell(u,v)}$. 
 \begin{equation}
 \label{eq:weight:edge}
 w(x, y) = \beta^{-\ell(x,y)} \text{ for every edge } (u,v) \in E.
 \end{equation}
 
 Thus, the weight of an edge decreases exponentially with its level. It is easy to check that if Invariant~\ref{inv:shadow:level} holds, then assigning the weights to the edges in this manner  satisfies condition (1) of Definition~\ref{def:structure}.

\begin{corollary}
\label{cor:inv:shadow:level}
Suppose that Invariant~\ref{inv:shadow:level} holds and  edges are assigned  weights as in equations~\ref{eq:level:edge},~\ref{eq:weight:edge}. Then for every edge $(x, y) \in E$ we have: 
$$\beta^{-\max(\ell(x), \ell(y))-1} \leq w(x,y) \leq \beta^{-\max(\ell(x), \ell(y))+1}.$$ 
\end{corollary}

\begin{proof}
Since each shadow-level differs from the actual level by at most one (see Invariant~\ref{inv:shadow:level}), the maximum value among the shadow-levels  also differs from the maximum value among the actual levels by at most one. Specifically, we get: $\max(\ell(x), \ell(y)) - 1 \leq \ell(x,y) = \max(\ell_x(x, y), \ell_y(x, y)) \leq \max(\ell(x), \ell(y))+1$. The corollary now follows from the fact that the weight of an edge $(x, y) \in E$ is given by $w(x, y)  = \beta^{- \ell(x, y)}$. 
\end{proof}

As in Section~\ref{sec:overview}, we now define the concept of an edge {\em marked} by a node. Consider any edge $(x, y) \in E$ incident to a node $y \in V$. If $\ell_y(x, y) = \ell(y)+1$, then we say that the edge $(x, y)$ has been {\em up-marked} by the node $y$. Similarly, if $\ell_y(x, y) = \ell(y) - 1$, then we say that the edge $(x, y)$ has been down-marked by the node $y$. And if $\ell_y(x, y) = \ell(y)$, then we say that the edge $(x, y)$ is {\em un-marked} by the node $y$. We let $M_{up}(y)$ and $M_{down}(y)$ respectively be the set of all edges $(x, y) \in E$ incident to $y$ that have been up-marked and down-marked by $y$. For every  $i \in [K,  L]$, we let $E_i(y)$ be the set of all  edges $(x, y) \in E$ incident to $y$ that are at level $\ell(x, y) = i$. 
\begin{eqnarray}
\label{eq:mark:up}
M_{up}(y) &  = &  \{ (x, y) \in E : \ell_y(x, y) = \ell(y) + 1\} \\
\label{eq:mark:down}
\qquad  M_{down}(y) & = &  \{ (x, y) \in E : \ell_y(x, y) = \ell(y) - 1 \}  \\
\label{eq:mark:level}
E_i(y) & = &  \{ (x, y) \in E : \ell(x, y) = i \} 
\end{eqnarray}

\subsection{Different states of a node.}
\label{sub:sec:node:state}

Our goal is to maintain a nice-partition in $G$. In Section~\ref{sub:sec:shadow-level}, we defined the concept of shadow-levels so as to ensure that the  edge-weights satisfy condition (1) of Definition~\ref{def:structure}. In this section, we present  a framework which will ensure that the {\em node-weights} satisfy the remaining conditions (2), (3) of Definition~\ref{def:structure}. Towards this end, we first need to define the concept of an {\em activation} of a node.

\medskip
\noindent {\bf Activations of a node.} The deletion of an edge $(x, y)$ in $G$ leads to a decrease in the values of  $W_x$ and $W_y$. In contrast, when an edge $(x, y)$ is inserted in $G$, we  assign  values to its two shadow-levels $\ell_x(x, y)$ and  $\ell_y(x, y)$ in such a way  that Invariant~\ref{inv:shadow:level} holds, and then assign a weight to the edge as per equation~\ref{eq:weight:edge}. This leads to an increase in the values of  $W_x$ and $W_y$. These two  events are called {\em natural activations} of the  endpoints $x, y$. In other words,  a node is {\em naturally activated} whenever an edge incident to it is either inserted into or deleted from  $G$.  The weight of a node changes whenever it encounters a natural activation.  Hence, such an event  might lead to a scenario where the node-weight becomes either too large or too small, thereby violating either condition (2) or condition (3) of Definition~\ref{def:structure}. For example, consider a node $y$ at a level $\ell(y) > K$ whose current weight is just slightly smaller than one. Thus, we have: $1- \delta \leq W_y < 1$ for some  small $\delta$. Now, suppose that $y$ gets naturally activated due to the insertion of an  edge $(x, y)$. Further, suppose that this leads to the value of $W_y$ becoming larger than one after the natural activation. So the node $y$ violates condition (2) of Definition~\ref{def:structure}. In our algorithm, at this stage the node $y$ will  select some edge $(x', y) \in E_{\ell(y)}(y)$ and up-mark that edge. Specifically, the node  will set $\ell_y(x', y) \leftarrow \ell(y) +1$,  insert the edge $(x', y)$ into the sets $M_{up}(y)$ and $E_{\ell(y)+1}(y)$, and remove the edge from the set $E_{\ell(y)}(y)$. The new level of the edge will be given by $\ell(x', y) = \ell(y) + 1$. This will reduce the node-weight $W_y$ by $\beta^{-\ell(y)} - \beta^{-(\ell(y) + 1)}$, and (hopefully) the new value of $W_y$ will again be smaller than one. The up-marking of the edge $(x', y)$, however, will  change  the weight of the other endpoint $x'$. We call such an event an {\em induced activation} of $x'$. Specifically, an {\em induced activation} of a node $x'$ refers to the event when the node-weight $W_{x'}$ increases (resp. decreases) because the other endpoint $y$ of an incident edge $(x', y)$  has decreased (resp. increased) its shadow-level $\ell_y(x', y)$. 

In general, consider an activation of a node $y$ that increases its weight.  Suppose that the node  wants  to revert this change (weight increase) so as to ensure that conditions (2) and (3) of Definition~\ref{def:structure}  remain satisfied. Then it either up-marks some edges from $E_{\ell(y)}(y)$ or un-marks some edges from $M_{down}(y)$. This, in turn, might activate some of the neighbours of $y$. 

Similarly, consider an activation of a node $y$ that decreases its weight. Suppose that the node wants  to revert this change (weight decrease) so as to ensure that conditions (2) and (3) of Definition~\ref{def:structure}  remain satisfied. Then it either down-marks some edges from $E_{\ell(y)}(y)$ or un-marks some edges from $M_{up}(y)$. Again, this might in turn activate  some of the neighbours of $y$. 

We require that a node cannot simultaneously have an up-marked and a down-marked edge incident on it. This requirement is formally stated in  Invariant~\ref{main:inv:up:down}. Intuitively,  a node has up-marked incident edges when it is trying to ensure that its weight does not become too large, and down-marked incident edges when it is trying to ensure that its weight does not become too small. Thus, it makes sense to assume that a node cannot simultaneously be in both these states. 

\begin{invariant}
\label{main:inv:up:down}
For every node $y \in V$, either $M_{up}(y) = \emptyset$ or $M_{down}(y) = \emptyset$.
\end{invariant}

Invariant~\ref{main:inv:shadow:level} states that if a node $y$ has up-marked or down-marked an incident edge $(x,y)$, then the shadow-level $\ell_x(x, y)$ of the other endpoint $x$  is no more than the level of $y$. Intuitively,  the node $y$  up-marks or down-marks an incident edge  only if it wants to change its weight $W_y$ without changing its own level $\ell(y)$. Suppose that the invariant is false, i.e., the node $y$ has up-marked or down-marked an edge $(x, y)$ with $\ell_x(x, y) > \ell(y)$. Then we have $\ell_y(x, y) \leq \ell(y) + 1 \leq \ell_x(x, y)$, where the first inequality follows from Invariant~\ref{inv:shadow:level}. But, this implies that  $y$ can {\em never change the weight  $w(x, y)$ by up-marking or down-marking $(x,y)$}, for the value of $w(x, y)$ is  determined by the shadow-level of the other endpoint $x$. Thus, the  node $y$ does not gain anything by  up-marking or down-marking the edge $(x, y)$. This is why we guarantee the following invariant.

\begin{invariant}
\label{main:inv:shadow:level}
For every edge $(x, y) \in E$, if $\ell_y(x, y) \neq \ell(y)$, then we must have $\ell_x(x, y) \leq \ell(y)$. 
\end{invariant}

\begin{table*}
\begin{large}
\begin{center}
  \begin{tabular}{ | l | c | c | r | r | } 
    \hline  
    &&&& \\ $\state[y]$& Weight-range & Up-marked  & Down-marked  & Other  \\ & & edges & edges & constraints \\ \hline
    &&&& \\ 1. $\Up$& $1-\frac{1}{\beta} \leq W_y < 1$  &  & $M_{down}(y) = \emptyset$ & $E_{\ell(y)}(y) \neq \emptyset$ \\ &&&& \\ \hline
         &&&& \\ 2. $\Down$& $f(\beta) \leq W_y < 1-\frac{2}{\beta}$ & $M_{up}(y) = \emptyset$ & & If $\ell(y) > K$, then   \\  & & & & $E_{\ell(y)}(y) - M_{down}(y) \neq  \emptyset$ \\ &&&& \\ \hline
                &&&& \\ 3. $\text{{\sc Slack}}$& $0 \leq W_y < f(\beta)$ & $M_{up}(y) = \emptyset$ & $M_{down}(y) = \emptyset$ & $\ell(y) = K$ \\ &&&& \\ \hline
    &&&& \\ 4. $\Idle$& $1-\frac{2}{\beta} \leq W_y < 1 - \frac{1}{\beta}$ & $M_{up}(y) = \emptyset$ & $M_{down}(y) = \emptyset$ &   \\ &&&& \\ \hline 
    &&&&\\ 5. $\Upb$& $1-\frac{2}{\beta} \leq W_y < 1 - \frac{1}{\beta}$ & $M_{up}(y) \neq \emptyset$ & $M_{down}(y) = \emptyset$ & \\ &&&& \\ \hline
    &&&& \\ 6. $\Downb$& $1-\frac{2}{\beta} \leq W_y < 1 - \frac{1}{\beta}$ & $M_{up}(y) = \emptyset$ & $M_{down}(y) \neq \emptyset$ &  \\ &&&& \\ \hline
  \end{tabular}
  \end{center}
  \end{large}
  \caption{\label{fig:different:states} Constraints satisfied by a node in different states.}
    \end{table*}

\noindent {\bf Six different states.}
For technical reasons, we will require that a node is always in one of six   possible {\em states}. See Table~\ref{fig:different:states}. It is easy to check that this is sufficient to ensure conditions (2), (3) of Definition~\ref{def:structure}. See Lemma~\ref{lm:different:states}.
One way to classify  these states is as follows. Definition~\ref{def:structure} requires that the weight of a node $y$ lies in the range $0 \leq W_y < 1$. We partition this range into four intervals: $I_1, I_2, I_3$ and $I_4$.  These intervals are non-empty as long as $\beta$ is a sufficiently large constant. 
\begin{eqnarray*}
I_1 = \left[0,f(\beta)\right), \ \  I_2 =  \left[f(\beta), 1-2/\beta\right) \\ 
 I_3 = \left[1-2/\beta, 1-1/\beta\right) \text{ and } I_4 =\left[1-1/\beta, 1\right).
 \end{eqnarray*}
A node $y$ is in $\Up$ state when $W_y \in I_4$, $\Down$ state when $W_y \in I_2$, and $\text{{\sc Slack}}$ state when $W_y \in I_1$. As per Table~\ref{fig:different:states},  the node $y$ has to satisfy some additional constraints when $\state[y] \in \{ \Up, \Down, \Slack\}$. Finally, if  $W_y \in I_3$, then  $y$ is in one of three possible states -- $\Idle$, $\Upb$, $\Downb$ -- depending on whether or not  it has   up-marked or down-marked any incident edge. By Invariant~\ref{main:inv:up:down}, a node cannot simultaneously up-mark some incident edges and down-mark some other incident edges. Hence,   three cases  can occur when $W_y \in I_3$. (a) $M_{up}(y) = M_{down}(y) = \emptyset$. In this case $y$ is in $\Idle$ state. (b) $M_{down}(y) = \emptyset$ and $M_{up}(y) \neq \emptyset$. In this case $y$ is in $\Upb$ state. (c) $M_{up}(y) = \emptyset$ and $M_{down}(y) \neq \emptyset$. In this case $y$ is in $\Downb$ state.

\smallskip The  six states are precisely defined in Table~\ref{fig:different:states}.

\begin{lemma}
\label{lm:different:states}
If a node $y \in V$ is in one of the states described in Table~\ref{fig:different:states}, then its weight $W_y$ satisfies conditions (2) and (3) of Definition~\ref{def:structure}. 
\end{lemma}

\begin{proof}
In every state, we have $0\leq W_y < 1$ (see Table~\ref{fig:different:states}). We  consider two mutually exclusive and exhaustive cases. (a) $f(\beta) \leq W_y < 1$. (b) $0 \leq W < f(\beta)$. In case (a), clearly the node-weight $W_y$ satisfies conditions (2), (3) of Definition~\ref{def:structure}. In case (b), the node must be in $\Slack$ state (see Table~\ref{fig:different:states}), and so we must have $\ell(y) = K$. Thus, the node satisfies conditions (2), (3)  of Definition~\ref{def:structure} even in case (b).
\end{proof}

Note that each of the intervals $I_1, I_2, I_3$ and $I_4$ defined above is of length at least $1/\beta$ (see equation~\ref{eq:beta:K:L}). On the other hand, for every edge $(u, v) \in E$ we have $w(u,v) \leq 1/\beta^K$, for $K$ is the minimum possible level in a nice-partition. Accordingly, a natural or induced activation of a node can change its weight by at most $1/\beta^K$. Note that $1/\beta^K$ is much smaller than $1/\beta$. This apparently simple observation has an important implication, namely, that a node must be activated at least $\beta^{K-1}$ times for its weight to cross the feasible range of any interval in $\{I_1, I_2, I_3, I_4\}$. As a corollary, if a node $y$ has, say, $W_y \in I_3$ just before getting activated, then the activation can only move $W_y$ to a {\em neighbouring} interval --  $I_2$ or $I_4$. But it is not possible to have $W_y \in I_3$ just before the activation, and $W_y \in I_1$ just after the activation. Throughout the rest of the paper, we will be using this observation each time we consider the effect of an activation on a node.
Next, we will briefly explain the motivation behind considering all these different states.

\smallskip
\noindent {\bf 1.} $\state[y] = \Up$. See row (1) in Table~\ref{fig:different:states}.

\noindent A node $y$ is in  $\Up$ state when $1-1/\beta \leq W_y < 1$. In this state the node's weight is close to one. Hence, whenever its weight increases further due to an activation the node tries to up-mark some incident edges from $E_{\ell(y)}(y)$, in the hope that this would reduce the node's weight and  ensure that  $W_y$ never exceeds one. The node $y$ can up-mark an edge only if  the set  $E_{\ell(y)}(y)$ is nonempty. Hence, we require that  $E_{\ell(y)}(y) \neq \emptyset$. Further, to ensure that a up-marking does not violate Invariant~\ref{main:inv:up:down}, we require that $M_{down}(y) = \emptyset$.   

\smallskip
\noindent {\bf 2.} $\state[y] =  \Down$. See row (2) in Table~\ref{fig:different:states}.

\noindent
A node $y$ is in  $\Down$ state when $f(\beta) \leq W_y < 1-2/\beta$. In this state the node's weight is close to the threshold $f(\beta)$. There are two cases to consider here, depending on the current level of the node.

\smallskip
\noindent
{\bf 2-a.} $\ell(y) > K$. In this case, whenever the value of $W_y$ decreases further due to an activation, the node tries to down-mark some incident edges from $E_{\ell(y)}(y) \setminus M_{down}(y)$, in the hope that this would increase the node's weight and  ensure that  $W_y$ does not drop below the threshold $f(\beta)$. The node $y$ can down-mark an edge only if the set  $E_{\ell(y)}(y) \setminus M_{down}(y)$ is nonempty. Hence, we require that  $E_{\ell(y)}(y) \setminus M_{down}(y) \neq \emptyset$.  Furthermore, in order  to ensure that a down-marking does not violate Invariant~\ref{main:inv:up:down}, we require that $M_{up}(y) = \emptyset$.  

\smallskip
\noindent {\bf 2-b.} $\ell(y) = K$. In this case, the node $y$ cannot down-mark any  incident edge $(x, y)$, for we must always have $\ell_y(x, y) \in [K, L]$. Thus, we get $M_{down}(y) = \emptyset$ in addition to the constraints specified in row (2) of Table~\ref{fig:different:states}. If an activation makes $W_y$  smaller than $f(\beta)$, then we simply  set $\state[y] \leftarrow \Slack$.

\smallskip
\noindent
We highlight one apparent discrepancy between the states $\Up$ and $\Down$. If a node $y$ is in $\Down$ state with $\ell(y) > K$, then it tries to down-mark some  edges from $E_{\ell(y)}(y) \setminus M_{down}(y)$ after an activation that reduces its weight. However, if the same node is in $\Up$ state, then it tries to up-mark some  edges from $E_{\ell(y)}(y)$  after an activation that increases its weight. This  discrepancy is due to the fact that  $E_{\ell(y)}(y) \cap M_{up}(y) = \emptyset$, as  every edge $(x, y ) \in M_{up}(y)$ has $\ell(x, y) \geq \ell_y(x, y) = \ell(y) +1$. In other words, an edge up-marked by $y$ can never belong to the set $E_{\ell(y)}(y)$, and hence $E_{\ell(y)}(y) \setminus M_{up}(y) = E_{\ell(y)}(y)$. In contrast, an edge $(x, y) \in M_{down}(y)$  belongs to the set $E_{\ell(y)}(x, y)$ if $\ell_x(x, y) = \ell(y)$.

\smallskip
\noindent {\bf 3.} $\state[y] = \text{{\sc Slack}}$. See row (3) in Table~\ref{fig:different:states}.

\noindent
A node $y$ is in  $\text{{\sc Slack}}$ state when $0 \leq W_y < f(\beta)$. In order to ensure condition (2) of Definition~\ref{def:structure}, we require that the node  be at level $K$. Since $K$ is the minimum possible level, there is no need for the node to prepare for moving down to a lower level in future. Hence, we require that $M_{down}(y) = \emptyset$. Further, the  node's weight  is currently so small that it will take quite some time before the node has to prepare for moving up to a higher level. Hence, we  require that $M_{up}(y) = \emptyset$.

\smallskip
\noindent {\bf 4.} $\state[y] = \Idle$. See row (4) in Table~\ref{fig:different:states}.

\noindent
 A node $y$ is in  $\Idle$ state when $1-2/\beta \leq W_y < 1-1/\beta$ and $M_{up}(y) = M_{down}(y) = \emptyset$. In this state the node's weight is neither too large nor too small, and the node does not have any up-marked or down-marked incident edges. Intuitively, the node need not worry even if its weight changes due to an activation in this state. In other words, when a node gets activated in $\Idle$ state, it does not up-mark, down-mark or un-mark any of its incident edges. After a sufficiently large number of activations when the node's weight drops below (resp. rises above) the threshold $1-2/\beta$ (resp. $1-1/\beta$), it  switches to the state $\Down$ (resp. $\Up$). 

\smallskip
\noindent {\bf 5.}  $\state[y] = \Upb$. See row (5) in Table~\ref{fig:different:states}. The term ``$\Upb$'' stands for ``{\sc Up-Backtrack}''.

 \noindent
 A node $y$ is in  $\Upb$ state when $1-2/\beta \leq W_y < 1-1/\beta$, $M_{up}(y) \neq \emptyset$ and $M_{down}(y) = \emptyset$. Intuitively, this state of the node captures the following scenario. Some time back the node $y$ was in $\Up$ state with $1-1/\beta \leq W_y < 1$,  $M_{up}(y) \neq \emptyset$ and $M_{down}(y) = \emptyset$. From that point onward, the node encountered a large number of activations that kept on reducing its weight. Eventually, the value of $W_y$ became smaller than $1-1/\beta$ and the node entered the state $\Upb$. If the node keeps  getting activated in this manner, then in near future $W_y$ will become smaller than $1-2/\beta$ and the node $y$ will have to enter the state $\Down$. At that time we must have $M_{up}(y) = \emptyset$. In other words, the node $y$ has to ensure that $M_{up}(y) = \emptyset$ before its weight drops below the threshold $1-2/\beta$. Thus, whenever $\state[y] = \Upb$ and the node-weight $W_y$ decreases due to an activation, the node $y$ un-marks some edges from $M_{up}(y)$.

\smallskip
\noindent {\bf 6.} $\state[y] = \Downb$. See row (6) in Table~\ref{fig:different:states}. The term ``$\Downb$'' stands for ``{\sc Down-Backtrack}''. 

\smallskip
\noindent
A node $y$ is in  $\Downb$ state when $1-2/\beta \leq W_y < 1-1/\beta$, $M_{down}(y) \neq \emptyset$ and $M_{up}(y) = \emptyset$. Intuitively, this state of the node captures the following scenario. Some time back the node $y$ was in $\Down$ state with $f(\beta) \leq W_y < 1-2/\beta$,  $M_{down}(y) \neq \emptyset$ and $M_{up}(y) = \emptyset$. From that point onward, the node encountered a large number of activations that kept on increasing its weight. Eventually, the value of $W_y$ became greater than $1-2/\beta$ and the node entered the state $\Downb$. If the node keeps  getting activated in this manner, then in near future $W_y$ will become greater than $1-1/\beta$ and it will have to enter the state $\Up$. At that time we must have $M_{down}(y) = \emptyset$. In other words, the node $y$ has to ensure that $M_{down}(y) = \emptyset$ before its weight increases beyond the threshold $1-1/\beta$. Thus, whenever $\state[y] = \Downb$ and  $W_y$ increases due to an activation, the node $y$ un-marks some edges from $M_{down}(y)$.

\subsection{Dirty nodes.}
\label{sub:sec:dirty:node}
Our algorithm   maintains a bit $D[y] \in \{0, 1\}$ associated with each node $y \in V$. We say that the node $y$ is {\em dirty} if $D[y] = 1$ and {\em clean} otherwise. Intuitively, the node $y$ is dirty when it is unsatisfied about its current condition and it wants to up-mark, down-mark or un-mark some of its incident edges. Once a dirty node is done with up-marking, down-marking or un-marking the relevant edges, it becomes clean again. 

In our algorithm, a node becomes dirty only after it encounters a natural or induced activation. The converse of this statement, however, is not true.  There may be times when a node remains clean even after getting activated, and this will be crucial in bounding the worst-case update time of our algorithm. Whether or not a node will become dirty due to an activation depends on: (1) the  state of the node,  (2) the type of the activation under consideration (whether it increases or decreases the node-weight), and (3) the node's current level. We  have three rules that determine when a node becomes dirty.

\begin{Rule}
\label{rule:dirty:up}
A node $y$ with $\state[y] \in \{\Up, \Downb\}$ becomes dirty after  an activation that increases its weight. In contrast, such a node {\em does not} become dirty after an activation that decreases its weight.
\end{Rule}

\noindent {\em Justification for Rule~\ref{rule:dirty:up}.}

\smallskip
\noindent {\em Case 1.} $\state[y] = \Up$. Here, we have  $1-1/\beta \leq W_y < 1$ and $M_{down}(y) = \emptyset$. If an activation increases the value of $W_y$, then $y$ needs to up-mark some  edges from $E_{\ell(y)}(y)$, in the hope  that  $W_y$ remains smaller than $1$ (see the discussion in Section~\ref{sub:sec:node:state}). Hence,  the node  becomes dirty. In contrast, if an activation reduces the value of $W_y$, then $y$ need not  up-mark, down-mark or un-mark any of its incident edges. Due to this inaction, if it so happens that  $1-2/\beta \leq W_y < 1-1/\beta$ after the activation, then the node simply switches   to state $\Upb$ or $\Idle$ depending on whether or not $M_{up}(y) \neq \emptyset$.

\smallskip
\noindent {\em Case 2.} $\state[y] = \Downb$. Here, we have $1-2/\beta \leq W_y < 1-1/\beta$, $M_{down}(y) \neq \emptyset$ and $M_{up}(y) = \emptyset$. Such a node must un-mark all its incident edges before  its weight rises past the threshold $1-1/\beta$ (see the discussion in Section~\ref{sub:sec:node:state}). Hence, whenever its weight increases due to an activation and   $\state[y] = \Downb$, the node $y$ becomes dirty and un-marks some edges from $M_{down}(y)$. In contrast, if an activation reduces its weight, then the node $y$ need not  up-mark, down-mark or un-mark any of its incident edges. Due to this inaction, if it so happens that $f(\beta) \leq W_y < 1-2/\beta$ after the activation, then we  set  $\state[y] \leftarrow \Down$. At this point, if we have $E_{\ell(y)}(y) \setminus M_{down}(y) = \emptyset$ and $\ell(y) > K$, then the node $y$ moves  to a lower level while being in $\Down$ state (see Case 2-b in Section~\ref{sub:sec:update:status}).

\begin{Rule}
\label{rule:dirty:down}
Consider a  node $y$ such that either (1) $\state[y] = \Down$ and $\ell(y) > K$, or (2) $\state[y] = \Upb$. This node  becomes dirty after  an activation that decreases its weight. In contrast, the node  does not become dirty after an activation that increases its weight.
\end{Rule}

\noindent {\em Justification for Rule~\ref{rule:dirty:down}.}

\smallskip
\noindent {\em Case 1.} $\state[y] = \Down$ and $\ell(y) > K$. Thus, we have  $f(\beta) \leq W_y < 1-2/\beta$ and $M_{up}(y) = \emptyset$. If an activation decreases its weight, then $y$ needs to down-mark some  edges from $E_{\ell(y)}(y) \setminus M_{down}(y)$, in the hope that $W_y$ does not become smaller than $f(\beta)$ (see the discussion in Section~\ref{sub:sec:node:state}). Hence,  the node $y$ becomes dirty. In contrast, if an activation increases its weight, then $y$ need not  up-mark, down-mark or un-mark any of its incident edges. Due to this inaction, if it so happens that  $1-2/\beta \leq W_y < 1-1/\beta$ after the activation, then the node simply switches   to state $\Downb$ or $\Idle$ depending on whether or not $M_{down}(y) \neq \emptyset$.

\smallskip
\noindent {\em Case 2.} $\state[y] = \Upb$. Thus, we have $1-2/\beta \leq W_y < 1-1/\beta$, $M_{up}(y) \neq \emptyset$ and $M_{down}(y) = \emptyset$. Such a node must un-mark all its incident edges before  its weight drops below the threshold $1-2/\beta$ (see the discussion in Section~\ref{sub:sec:node:state}).  Hence, whenever its weight decreases due to an activation and   $\state[y] = \Upb$, the node $y$ becomes dirty and un-marks some edges from $M_{up}(y)$. In contrast, if an activation increases its weight, then $y$ need not up-mark, down-mark or un-mark any of its incident edges. Due to this inaction, if it so happens  that $1- 1/\beta \leq W_y < 1$ after the activation, then we set  $\state[y] \leftarrow \Up$. At this point, if we have $E_{\ell(y)}(y) = \emptyset$, then the node $y$ moves to a higher level while being in $\Up$ state (see Case 2-a in Section~\ref{sub:sec:update:status}).

\begin{Rule}
\label{rule:dirty:slack}
A  node $y$ with either (1) $\state[y] \in \{\Slack, \Idle\}$ or (2) \{$\state[y] = \Down$ and $\ell(y) = K$\} never becomes dirty after an activation.
\end{Rule}

\noindent {\em Justification for Rule~\ref{rule:dirty:slack}.}

\smallskip
\noindent {\em Case 1.} $\state[y] = \Slack$. Here, we have  $0 \leq W_y < f(\beta)$, $M_{up}(y) = M_{down}(y) = \emptyset$ and $\ell(y) = K$. When such a node gets activated,  it need not  up-mark or down-mark  any of its incident edges. Due to this inaction,  if it so happens that  $f(\beta) \leq W_y < 1-2/\beta$ after the activation, then we set $\state[y] \leftarrow \Down$.

\smallskip
\noindent {\em Case 2.} $\state[y] = \Idle$. Here, we have $1-2/\beta \leq W_y < 1-1/\beta$ and $M_{up}(y) = M_{down}(y) =  \emptyset$. When such a node gets activated, it need not  up-mark or down-mark  any of its incident edges. Due to this inaction, if it so happens that $1-1/\beta \leq W_y < 1$ after the activation, then we set $\state[y] \leftarrow \Up$. At this point, if we have $E_{\ell(y)}(y) = \emptyset$, then the node $y$ moves to a higher level while being in $\Up$ state (see Case 2-a in Section~\ref{sub:sec:update:status}). 
In contrast, if it so happens that $f(\beta) \leq W_y < 1-2/\beta$ after the activation, then we set $\state[y] \leftarrow \Down$.  At this point, if we have $E_{\ell(y)}(y) \setminus M_{down}(y) = \emptyset$ and $\ell(y) > K$, then the node $y$ moves  to a lower level while being in $\Down$ state (see Case 2-b in Section~\ref{sub:sec:update:status}).

\smallskip
\noindent {\em Case 3.} $\state[y] = \Down$ and $\ell(y) = K$. Thus, we have $f(\beta) \leq W_y < 1 - 2/\beta$ and $M_{up}(y) = \emptyset$. Since $\ell(y) = K$ and $\ell_y(x, y) \in [K, L]$ for every edge $(x, y) \in E$, we also have $M_{down}(y) = \emptyset$. When such a node gets activated, it need not up-mark or down-mark any of its incident edges. Due to this inaction, if we have $0 \leq W_y < f(\beta)$ after the activation, then we set $\state[y] \leftarrow \Slack$. In contrast, if  $1-2/\beta \leq W_y < 1 - 1/\beta$ after the activation, then we set $\state[y] \leftarrow \Idle$.

\begin{corollary}
\label{cor:activation}
If an activation of a node $y$ makes it dirty, then the state of the node remains the same just before and just after the activation (see Section~\ref{sub:sec:recap}).
\end{corollary}

\begin{proof}
While justifying Rules~\ref{rule:dirty:up}~--~\ref{rule:dirty:slack}, whenever we changed the state of the node $y$ due to an activation, we ensured that the node did not become dirty.
\end{proof}

\subsection{Data structures.}
\label{sub:sec:data:structures}
In our dynamic algorithm, every node $y \in V$ maintains the following data structures.

\smallskip
\noindent 1. Its weight $W_y$, level $\ell(y)$, and state $\text{{\sc State}}[y] \in \{ \text{{\sc  Up, Down, Slack, Idle, Up-B, Down-B}}\}$.

\smallskip
\noindent 2. The sets  $M_{up}(y), M_{down}(y)$ as balanced search trees. 

\smallskip
\noindent 3. A  bit $D[y] \in \{0, 1\}$ to indicate if the node $y$ is {\em dirty}. 

\smallskip
\noindent 4. For every level $i \in \{0, \ldots, L\}$, the set of edges $E_i(y)$ as a balanced search tree.

\smallskip
\noindent
Furthermore, every edge $(x, y) \in E$ maintains the values of its weight $w(x, y)$ and level $\ell(x, y)$.

\paragraph{Remark about maintaining the shadow-levels.}
Note that we do not {\em explicitly} maintain the shadow-level $\ell_x(x, y)$ of a node $y \in V$ with respect to an  edge $(x, y) \in E$. This is due to  the following reason. 

For the sake of contradiction, suppose that our algorithm in fact maintains the values of the shadow-levels $\ell_y(x, y)$. Consider a scenario where the node $y$ has $\state[y] = \Up$, $\ell(y) = i$, and the value of $W_y$ is very close to one. Next, suppose that an activation increases the value of $W_y$, and the node up-marks one or more edges from the set $E_{i}(y)$ to ensure that the value of $W_y$ remains smaller than one. Since $\ell(x, y) = i + 1$ for every edge $(x, y) \in M_{up}(y)$, all the newly up-marked edges  get deleted from the set $E_{i}(y)$ and added to the set $E_{i+1}(y)$. At this point, we might end up in a situation where $E_{i}(y) = \emptyset$, which violates a constraint  of row (1) in Table~\ref{fig:different:states}. Our algorithm  deals with this issue by moving the node $y$ up to level $(i+1)$, i.e., by setting $\ell(y) \leftarrow (i+1)$. Since $E_i(y) = \emptyset$, this does not affect the weight of any edge. However, for every edge $(x, y) \in E$ with $\ell(x,y) > i+1$, the shadow-level $\ell_y(x,y)$  changes from $i$ to $(i+1)$. Since each edge $(x, y) \in E$ with $\ell(x, y) > i+1$ has weight at most $\beta^{-(i+2)}$, and since $W_y < 1$, there can be $\beta^{i+2}-1$ many such edges. Accordingly,  the node $y$ might be forced to change the values of the shadow-levels $\ell_y(x, y)$ for $O(\beta^{i+2})$ many edges $(x, y)$. The worst-case update time  then becomes $O(\beta^{i+2})$, which is polynomial in $n$ for large values of $i$. 

We   avoid this problem by  giving up on explicitly maintaining the values of the shadow-levels $\ell_y(x, y)$. Still  we can determine the value of $\ell_x(x, y)$ in $O(\log n)$ time from the  data structures that  {\em are} in fact maintained by us. Specifically, we know that if $(x, y) \in M_{up}(y)$, then $\ell_y(x, y) = \ell(y) + 1$. Else if $(x, y) \in M_{down}(y)$, then $\ell_y(x, y) = \ell(y) - 1$. Finally, else if $(u, y) \notin M_{up}(y) \cup M_{down}(y)$, then $\ell_y(x, y) = \ell(y)$.

For ease of exposition,     we  nevertheless use the notation $\ell_y(x,y)$ while describing  our algorithm in subsequent sections.  Whenever we do this, the reader should keep it in mind that we are implicitly computing  $\ell_y(x, y)$ as per the above procedure.

\section{Some basic subroutines}
\label{sec:subroutines}

\subsection{The subroutine UPDATE-STATUS$(y)$.}
\label{sub:sec:update:status}

This subroutine is called each time a node $y$ experiences a natural or an induced activation. This  tries to ensure, by changing the state and level of $y$ if necessary, that $y$ satisfies the constraints specified in Table~\ref{fig:different:states}. If the subroutine fails to ensure this condition, then our algorithm HALTS. During the analysis of our algorithm, we will prove that it never HALTS due to a call to UPDATE-STATUS$(y)$. This implies that every node satisfies the constraints in Table~\ref{fig:different:states}, and hence  Lemma~\ref{lm:different:states} guarantees that conditions (2) and (3) of Definition~\ref{def:structure} continue to remain satisfied all the time.

We say that a node is {\em fit} in a  state $X \in \{ \Up, \Down, \Slack, \Idle, \Upb, \Downb \}$ if it satisfies all the constraints for state $X$ as specified in Table~\ref{fig:different:states}, and {\em unfit} otherwise.  
If $D[y] = 1$, then our algorithm HALTS if $y$ is unfit in its {\em current} state. In contrast, if $D[y] = 0$, then our algorithm HALTS if $y$ is unfit in {\em every} state, albeit with one caveat: If the node is unfit in either state $\Up$ or state $\Down$, then we first try to make it fit in that state by changing its level $\ell(y)$.  
Hence, there is a sharp distinction between the treatments received by the   clean nodes on the one hand and the dirty nodes on the other.  Specifically, the state of a node $y$ can change during to a call to UPDATE-STATUS$(y)$ only if $y$ is clean at the beginning of the call. This distinction  comes from Corollary~\ref{cor:activation}, which requires that a node does not change its state if it becomes dirty. We now describe the subroutine in  details.

\paragraph{Case 1.} $D[y] = 1$. The node $y$ is dirty.

\noindent If $y$ is fit in its current state, then we terminate the subroutine. Otherwise our algorithm HALTS.

\paragraph{Case 2.} $D[y] = 0$.  The node $y$ is clean.

\noindent If  we can find some state  $X \in \{\Up, \Down, \Slack, \Idle, \Upb,\Downb\}$ in which  $y$ is fit, then we set $\state[y] \leftarrow X$ and terminate the subroutine. Else if the node $y$ is unfit in every state, then  we consider the  sub-cases 2-a, 2-b and 2-c.

\paragraph{Case 2-a.}  The node  is unfit in state $\Up$ only due to the last constraint in row (1) of Table~\ref{fig:different:states}. Thus, we have $1-1/\beta \leq W_y < 1$, $M_{down}(y) = \emptyset$ and $E_{\ell(y)}(y) = \emptyset$. Let $i \leftarrow \ell(y)$ be the current level of $y$. 
We find the minimum level $j > i$ where $E_j(y) \neq \emptyset$. Such a level $j$ must exist since $W_y > 0$. We move the node $y$ up to level $j$ by setting $\ell(y) \leftarrow j$. This does not change the weight of any edge. Furthermore, when the node was in level $i$, we had $\ell_y(x, y) \leq \ell(y) + 1 \leq i+1 \leq j$ for every edge $(x, y) \in E$ incident on $y$ (see Invariant~\ref{inv:shadow:level}). Hence, after the node moves up to level $j$, we have $\ell_y(x, y) =  j = \ell(y)$ for every edge $(x, y) \in E$. In other words, the node $y$ is not supposed to have any up-marked edges incident on it just after moving to level $j$. Accordingly, we set $M_{up}(y) \leftarrow \emptyset$. Then we terminate the subroutine. 

\paragraph{Case 2-b.}  The node  is unfit in state $\Down$ only due to the last constraint in row (2) of Table~\ref{fig:different:states}. Thus, we have $f(\beta) \leq W_y < 1-2/\beta$, $M_{up}(y) = \emptyset$, $E_{\ell(y)}(y) \setminus M_{down}(y) = \emptyset$ and $\ell(y) > K$. Let $i \leftarrow \ell(y)$ be the current level of $y$. 
We first move the node down to level $i-1$ by setting $\ell(y) \leftarrow i-1$. We claim that this does not change the level (and weight) of any edge. To see why the claim is true, consider any edge $(x, y) \in E$ incident on $y$. Since $M_{up}(y) = \emptyset$, we must have $\ell_y(x, y) \leq i$ just before the node moves down to level $(i-1)$. If $\ell_x(x, y) \geq i$, then the value of $\ell(x, y)$ is determined by the other endpoint $x$ and the level of such an edge does not change as $y$ moves down to level $(i-1)$. In contrast, if $\ell_x(x, y) < i$, then we have $(x, y) \in M_{down}(y)$: for otherwise  the edge $(x, y)$ will belong to the  set $E_i(y) \setminus M_{down}(y)$ which we have assumed to be empty. The level of such an edge remains equal to $(i-1)$ as the node $y$ moves down from level $i$ to level $(i-1)$. This concludes the proof of the claim that the edge-weights do not change as  $y$ moves down from level $i$ to level $(i-1)$. 
Next, consider any edge $(x, y)$ that was down-marked when the node $y$ was at level $i$. At that time, we had $\ell_y(x, y) = i-1$. Hence, after the node moves down to level $i-1$, we get $\ell_y(x, y) = i-1 = \ell(y)$. Thus, the node cannot have any down-marked edge incident on it just after moving down to level $i-1$. Accordingly, we set $M_{down}(y) \leftarrow \emptyset$. 
At this point, if we find that $E_{i-1}(y) = \emptyset$, then we move the node further down to the lowest level $K$, by setting $\ell(y) \leftarrow K$. This does not change the level and weight of any edge in the graph.
Finally, we terminate the subroutine. 

\paragraph{Case 2-c.} In every scenario other than  2-a and  2-b  described above, our algorithm HALTS.

\smallskip
\noindent{\bf A note on the space complexity.} In cases 2-a and 2-b of the above procedure, there is a step where we set $M_{up}(y) \leftarrow \emptyset$ and $M_{down}(y) \leftarrow \emptyset$ respectively. It is essential to execute this step in $O(\text{poly} \log n)$  time: otherwise we cannot claim that the update time of our algorithm is $O(\text{poly} \log n)$ in the worst-case. Unfortunately for us,  there can be $\Omega(\beta^{\ell(y)})$ many edges in the set $M_{up}(y)$ or $M_{down}(y)$. Hence, it will take $\Omega(\beta^{\ell(y)})$ time to empty that  set if we have to delete all those edges from the corresponding balanced search tree. Note that  $\beta^{\ell(y)} = \Omega(n)$ for large  $\ell(y)$. 

To address this concern, we  maintain two pointers $root[M_{up}(y)]$ and $root[M_{down}(y)]$ for each node $y \in V$. They respectively point to the root of the balanced search tree for $M_{up}(y)$ and $M_{down}(y)$.  When we want to set $M_{up}(y) \leftarrow \emptyset$ or $M_{down}(y) \leftarrow \emptyset$, we respectively set $root[M_{up}(y)] \leftarrow \text{NULL}$ or $root[M_{down}(y)] \leftarrow \text{NULL}$.  This takes only constant time. The downside of this approach is that the algorithm now uses up a lot of {\em junk space} in  memory: This space is occupied by the balanced search trees that were {\em emptied}  in the past.  As a result, the space complexity of the algorithm becomes $O(t  \text{ poly}\log n)$ for handling a sequence of $t$ edge insertions/deletions starting from an empty graph. This is due to the fact that our algorithm will be shown to have a worst-case update time of $O(\text{poly} \log n)$. Hence, we can upper bound the total time taken to handle these edge insertions/deletions by $O(t \text{ poly}\log n)$,  and this, in turn, gives a trivial upper bound on the amount of {\em junk space} used up in the memory. 

A standard way to bring down the space complexity is to run a {\em clean-up} algorithm {\em in the background}. Each time an edge is inserted into or deleted from the graph, we visit $O(\text{poly}\log n)$ memory cells that are currently junk and {\em free them up}. Thus, the worst case update time of the clean-up algorithm is also $O(\text{poly}\log n)$, and this  increases the overall update time of our scheme by only a $O(\text{poly} \log n)$ factor.  The  size of all sets $M_{up}(.)$ and $M_{down}(.)$ that exist at a given point in time is $O(m)$. Hence, this clean-up algorithm is at most $O(m)$ space ``behind'', i.e., the additional space requirement for junk space is $O(m)$.
  For ease of exposition, from this point onward we will simply assume that we can empty a balanced search tree in $O(1)$ time.

\begin{lemma}
\label{lm:time:update:status}
The subroutine {\em UPDATE-STATUS}$(y)$ takes  $O(\log n)$ time.
\end{lemma}

\begin{proof}
Case 1 can clearly be implemented in $O(1)$ time. In case 2-a,  we have to find the minimum level $j > i$ where $E_j(y) \neq \emptyset$. This operation takes time proportional to the number of levels, which is $L-K+1 = O(\log n)$. Everything else takes $O(1)$ time. Finally, case 2-b and case 2-c also take $O(1)$ time. 
\end{proof}

\subsection{The subroutine PIVOT-UP$(v, (u,v))$.}
\label{sub:sec:pivot:up}

This is described in Figure~\ref{new:fig:pivot:up}. This subroutine is called when the node $v$ is dirty and it wants to increase its shadow-level $\ell_v(u, v)$ with respect to the edge $(u, v)$.  There are two situations under which such an event can take place: (1)  $\state[v] = \Up$ and  $v$ wants to up-mark the edge $(u, v)$, and (2) $\state[v] =  \Downb$ and  $v$ wants to un-mark the edge $(u, v)$. The subroutine PIVOT-UP$(v, (u, v))$ updates the relevant data structures, decides whether the node $u$ should become dirty because of this event, and returns {\sc True} if the event changes the weight of the edge $(u,v)$ and {\sc False} otherwise.  Thus, if the subroutine returns {\sc True}, then this amounts to an induced activation of the node $u$.

\paragraph{The subroutine MOVE-UP$(v, (u,v))$. } Step (01) in Figure~\ref{new:fig:pivot:up}  calls another subroutine MOVE-UP$(v, (u, v))$. This subroutine (described in  Figure~\ref{new:fig:move:up})  updates the relevant data structures as the value of $\ell_v(u,v)$ increases by one, and returns {\sc True} if the weight $w(u, v)$  gets changed and {\sc False} otherwise.
To see an example where  MOVE-UP$(v, (u, v))$ returns {\sc False}, consider a situation where $\state[v] = \Downb$, $\ell(v) = i$, $\ell_v(u, v) = i-1$, and $\ell_u(u, v) = \ell(u) = i$. In this instance, even after the node $v$ increases the value of $\ell_v(u, v)$ by un-marking the edge  $(u, v)$, the weight $w(u, v)$ does not change.

The subroutine MOVE-UP$(v, (u, v))$ ensures that Invariant~\ref{main:inv:shadow:level} remains satisfied. Specifically, after   the value of $\ell_v(u, v)$  increases  we might have $\ell_v(u, v) > \ell(u)$, and then we must ensure that the edge $(u, v) \notin M_{up}(u) \cup M_{down}(u)$: otherwise Invariant~\ref{main:inv:shadow:level} will be violated (set $y = u$ and $x = v$ in Invariant~\ref{main:inv:shadow:level}). If we end up in this situation, then the subroutine MOVE-UP$(v, (u,v))$ removes  the edge  from $M_{up}(u) \cup M_{down}(u)$.


\begin{figure}[htbp]
\centerline{\framebox{
\begin{minipage}{5.5in}
\begin{tabbing}
01.   \= $Y \leftarrow \text{MOVE-UP}(v, (u, v))$  \qquad // See Figure~\ref{new:fig:move:up}. \\
02. \> {\sc If} \= $Y = \text{{\sc True}}$ and \big\{either  $\state[u] = \Upb$ or \\
 \>  \> $\left(\state[u] = \Down \text{ and } \ell(u) > K\right)$\big\}  \\  
03. \> \> \ \ \ \ \ \ \ \ \= $D[u] \leftarrow 1$ \\ 
04. \>  UPDATE-STATUS($u$) \\
05. \>  RETURN $Y$. 
\end{tabbing}
\end{minipage}
}}
\caption{\label{new:fig:pivot:up} PIVOT-UP($v, (u, v)$).}
\end{figure}

\begin{figure}[htbp]
\centerline{\framebox{
\begin{minipage}{5.5in}
\begin{tabbing}
01. \ \= $i_v \leftarrow \ell_v(u, v)$ \\
02. \>  $i_u \leftarrow \ell_u(u, v)$ \\
03. \> {\sc If} $(u, v) \in M_{down}(v)$ \\
04. \> \qquad \= $M_{down}(v) \leftarrow M_{down}(v) \setminus \{(u,v) \}$ \\
05. \> {\sc Else}  \\
06. \> \>  $M_{up}(v) \leftarrow M_{up}(v) \cup \{ (u, v) \}$ \\
07. \> {\sc If} $i_v + 1 > \ell(u)$  \\
08. \> \>  $M_{up}(u) \leftarrow M_{up}(u) \setminus \{ (u,v) \}$ \\
09. \> \>  $M_{down}(u) \leftarrow M_{down}(u) \setminus \{ (u,v) \}$ \\
10. \> {\sc If} $\ell(u, v) = \max(i_v+1, i_u)$ \\
11. \> \> RETURN {\sc False} \\
12. \> $w(u, v) \leftarrow \beta^{-(\ell(u,v)+1)}$ \\
13. \> $W_u \leftarrow W_u - \beta^{-\ell(u,v)} + \beta^{-(\ell(u,v)+1)}$ \\
14. \> $W_v \leftarrow W_v - \beta^{-\ell(u,v)} + \beta^{-(\ell(u,v)+1)}$ \\
15. \> $E_{\ell(u,v)}(u) \leftarrow E_{\ell(u,v)}(u) \setminus \{ (u, v) \}$ \\ 
16. \> $E_{\ell(u,v)+1}(u) \leftarrow E_{\ell(u,v)+1}(u) \cup \{(u,v) \}$  \\
17. \> $E_{\ell(u,v)}(v) \leftarrow E_{\ell(u,v)}(v) \setminus \{ (u, v) \}$ \\
18. \> $E_{\ell(u,v)+1}(v) \leftarrow E_{\ell(u,v)+1}(v) \cup \{(u,v) \}$ \\
19. \> $\ell(u,v) \leftarrow \ell(u,v) +1$ \\
20. \> RETURN {\sc True}
\end{tabbing}
\end{minipage}
}}
\caption{\label{new:fig:move:up} MOVE-UP($v, (u, v)$).}
\end{figure}

\paragraph{Deciding if the node $u$ becomes dirty.} We now continue with the description of the subroutine PIVOT-UP$(v, (u, v))$.  After step (01) in Figure~\ref{new:fig:pivot:up}, it remains to decide whether the node $u$ should become dirty. This decision is made following the three rules specified  in Section~\ref{sub:sec:dirty:node}. Note that if we increase the value of $\ell_v(u, v)$, then it can never lead to an increase in the weight $W_u$. Thus, if $Y = \text{{\sc True}}$, then it means that the weight $W_u$ dropped  during the call to the subroutine MOVE-UP$(v, (u,v))$.  On the other hand, if $Y = \text{{\sc False}}$, then it means that the weight $W_u$ did not change during the call to the subroutine MOVE-UP$(v, (u,v))$. In this event, the node $u$ never becomes dirty. 

As per  Rules~\ref{rule:dirty:up}~--~\ref{rule:dirty:slack}, if  the weight $W_u$ gets reduced, then  $u$ becomes dirty iff either  $\state[u] = \Upb$ or  $\left(\state[u] = \Down, \ell(u) > K\right)$. Thus, the subroutine sets $D[u] \leftarrow 1$ iff two conditions are satisfied: (1) $Y = \text{{\sc True}}$, and (2) either $\state[u] = \Upb$ or $\left( \state[u] = \Down, \ell(u) > K\right)$. 

Finally,  just before terminating the subroutine PIVOT-UP$(v, (u,v))$ in Figure~\ref{new:fig:pivot:up}, we  call the subroutine UPDATE-STATUS$(u)$. The reason  for this call is explained in the beginning of Section~\ref{sub:sec:update:status}.

\begin{lemma}
\label{lm:pivot:up}
The subroutine PIVOT-UP$(v, (u, v))$ takes $O(\log n)$ time. It returns {\sc True} if the weight $w(u, v)$ gets changed, and {\sc False} otherwise. The node $u$ becomes dirty only if the subroutine returns {\sc True}. 
\end{lemma}

\begin{proof}
A call to the subroutine UPDATE-STATUS$(y)$ takes $O(\log n)$ time, as per  Lemma~\ref{lm:time:update:status}. The rest of the proof follows from the description of the subroutine.
\end{proof}

\subsection{The subroutine PIVOT-DOWN$(v, (u, v))$.}
\label{sub:sec:pivot:down}

This is described in Figure~\ref{new:fig:pivot:down}. This subroutine is called when the node $v$ is dirty and it wants to decrease its shadow-level $\ell_v(u, v)$ with respect to the edge $(u, v)$.  There are two situations under which such an event can take place: (1)  $\state[v] = \Down$ and  $v$ wants to down-mark the edge $(u, v)$, and (2) $\state[v] =  \Upb$ and  $v$ wants to un-mark the edge $(u, v)$. The subroutine PIVOT-DOWN$(v, (u, v))$ updates the relevant data structures, decides whether the node $u$ should become dirty, and returns {\sc True} if  the weight of the edge $(u,v)$ gets changed and {\sc False} otherwise. Thus, if the subroutine returns {\sc True}, then this amounts to an induced activation of the node $u$. 
This subroutine, however, is {\em not} a mirror-image of the subroutine PIVOT-UP$(v, (u, v))$. The  difference between them is explained below.

In the subroutine PIVOT-DOWN$(v, (u,v))$, suppose that the node $v$ has decreased the value of $\ell_v(u, v)$, and this has increased the weight $W_u$. Furthermore, the node $u$ is currently in a state where  Rules~\ref{rule:dirty:up}~--~\ref{rule:dirty:slack} dictate that it should become dirty when its weight increases.   If this is the case, then the node $u$ attempts to {\em undo} its weight-change by increasing the value of $\ell_u(u,v)$. To take a concrete example, suppose that just before the subroutine PIVOT-DOWN$(v, (u,v))$ is called, we have $\state[v] = \Upb$, $\ell(v) = i$, $\ell_v(u,v) = i+1$, $\state[u] = \Up$, $\ell(u) = i$ and $\ell_u(u, v) = i$. The node $v$ now decreases the value of $\ell_v(u, v)$  by one, and un-marks the edge $(u, v)$. Thus, the weight $w(u, v)$ changes from $\beta^{-(i+1)}$ to $\beta^{-i}$. This also increases the weight $W_u$  by an amount $\beta^{-i} - \beta^{-(i+1)}$. The node $u$  will now undo  this change by up-marking the edge $(u, v)$, which will increase  $\ell_u(u, v)$ by one. This will bring the weight $W_u$ back to its initial value. In contrast,  the subroutine PIVOT-UP$(v, (u, v))$ does not allow the node $u$ to perform such ``undo'' operations. This ``undo'' operation performed by  $u$ in PIVOT-DOWN$(v, (u,v))$ will be crucial in bounding the  update time of our algorithm.

\paragraph{The subroutine MOVE-DOWN$(v, (u,v))$. } Step (01) in Figure~\ref{new:fig:pivot:down}  calls  MOVE-DOWN$(v, (u, v))$. This subroutine (described in Figure~\ref{new:fig:move:down}) updates the relevant data structures as the value of $\ell_v(u,v)$ decreases by one, and returns {\sc True} if the weight $w(u, v)$  gets changed and {\sc False} otherwise. 
To see an example where   MOVE-DOWN$(v, (u, v))$ returns {\sc False}, consider a situation where $\state[v] = \Down$, $\ell(v) = i$, $\ell_v(u, v) = i$, and $\ell_u(u, v) = \ell(u) = i$. In this instance, even after the node $v$ decreases the value of $\ell_v(u, v)$ by down-marking the edge  $(u, v)$, the weight $w(u, v)$ does not change.

\begin{figure}[htbp]
\centerline{\framebox{
\begin{minipage}{5.5in}
\begin{tabbing}
01.   \= $Y \leftarrow \text{MOVE-DOWN}(v, (u, v))$  \qquad // See Figure~\ref{new:fig:move:down}. \\
02. \> {\sc If} $\state[u] = \Up$  \\
03. \> \ \ \ \  \= {\sc If} $Y = \text{{\sc True}}$ \\
04. \> \> \ \ \ \  \= {\sc If} $(u,v) \notin M_{up}(u) \bigwedge \ell(u) \geq \ell_v(u, v)$ \\
05. \> \> \> \ \ \ \ \= MOVE-UP($u, (u,v)$) \\
06. \> \> \> \> UPDATE-STATUS($u$) \\
07. \> \> \> \> RETURN {\sc False} \\
08. \> \> \> {\sc Else} \\
09. \> \> \> \> $D[u] \leftarrow 1$ \\
10. \> \> \> \> UPDATE-STATUS($u$) \\
11. \> \> \> \> RETURN $Y$. \\
12. \> \> {\sc Else} \\
13. \> \> \> UPDATE-STATUS($u$) \\
14. \> \> \> RETURN $Y$. \\
15. \> {\sc Else if} $\state[u] = \Downb$ \\
16. \> \> {\sc If} $Y = \text{{\sc True}}$, {\sc Then} \\
17. \> \> \> {\sc If} $(u, v) \in M_{down}(u) \bigwedge \ell_v(u, v) < \ell(u)$  \\
18. \> \> \> \> MOVE-UP($u, (u,v)$) \\
19. \> \> \> \> UPDATE-STATUS($u$) \\
20. \> \> \> \> RETURN {\sc False} \\
21. \> \> \> {\sc Else} \\
22. \> \> \> \> $D[u] \leftarrow 1$ \\
23. \> \> \> \> UPDATE-STATUS($u$) \\
24. \> \> \> \> RETURN $Y$. \\
25. \> \> {\sc Else} \\
26. \> \> \> UPDATE-STATUS($u$) \\
27. \> \> \> RETURN $Y$. \\
28. \> {\sc Else} \\ 
29. \> \> UPDATE-STATUS($u$) \\
30. \> \> RETURN $Y$.
\end{tabbing}
\end{minipage}
}}
\caption{\label{new:fig:pivot:down} PIVOT-DOWN($v, (u, v)$). Steps (05) and (18) correspond to ``undo'' operations  by the node $u$.}
\end{figure}

\begin{figure}[htbp]
\centerline{\framebox{
\begin{minipage}{5.5in}
\begin{tabbing}
01. \  \= $i_v \leftarrow \ell_v(u, v)$ \\
02. \>  $i_u \leftarrow \ell_u(u, v)$ \\
03. \> {\sc If} $(u, v) \in M_{up}(v)$ \\  
04. \> \qquad \= $M_{up}(v) \leftarrow M_{up}(v) \setminus \{(u,v) \}$ \\
05. \> {\sc Else}  \\
06. \> \> $M_{down}(v) \leftarrow M_{down}(v) \cup \{ (u, v) \}$ \\
07. \> {\sc If} $\ell(u, v) = \max(i_v-1, i_u)$ \\
08. \> \> RETURN {\sc False} \\
09. \> $w(u, v) \leftarrow \beta^{-(\ell(u,v)-1)}$ \\
10. \> $W_u \leftarrow W_u - \beta^{-\ell(u,v)} + \beta^{-(\ell(u,v)-1)}$ \\
11. \> $W_v \leftarrow W_v - \beta^{-\ell(u,v)} + \beta^{-(\ell(u,v)-1)}$ \\
12. \> $E_{\ell(u,v)}(u) \leftarrow E_{\ell(u,v)}(u) \setminus \{ (u, v) \}$ \\
13. \>  $E_{\ell(u,v)-1}(u) \leftarrow E_{\ell(u,v)-1}(u) \cup \{(u,v) \}$  \\
14. \> $E_{\ell(u,v)}(v) \leftarrow E_{\ell(u,v)}(v) \setminus \{ (u, v) \}$ \\
15. \> $E_{\ell(u,v)-1}(v) \leftarrow E_{\ell(u,v)-1}(v) \cup \{(u,v) \}$ \\
16. \> $\ell(u,v) \leftarrow \ell(u,v) -1$ \\
17. \> RETURN {\sc True}
\end{tabbing}
\end{minipage}
}}
\caption{\label{new:fig:move:down} MOVE-DOWN($v, (u, v)$).}
\end{figure}

Unlike the subroutine MOVE-UP$(v, (u, v))$, here we need not worry about Invariant~\ref{main:inv:shadow:level} getting violated, for the following reason. Set $v = x$ and $u = y$ in Invariant~\ref{main:inv:shadow:level} just as we did while considering the subroutine MOVE-UP$(v, (u, v))$. If $\ell_u(u, v) =\ell(u)$, the Invariant~\ref{main:inv:shadow:level} clearly remains satisfied even as $\ell_v(u,v)$ decreases by one. If $\ell_u(u, v) \neq \ell(u)$, then by Invariant~\ref{main:inv:shadow:level} we have $\ell_v(u,v) \leq \ell(u)$ just before the call to MOVE-DOWN$(v, (u,v))$. In this case as well, Invariant~\ref{main:inv:shadow:level} continues to remain satisfied even as $\ell_v(u, v)$ decreases by one. 


\paragraph{Deciding if the node $u$   becomes dirty.}  We  continue with the description of  PIVOT-DOWN$(v, (u, v))$.  After step (01) in Figure~\ref{new:fig:pivot:down}, it remains to decide whether (a) the node $u$ is about to become dirty, and if the answer is yes, then whether (b)  the node can escape this fate by successfully executing an ``undo'' operation. Decision (a) is taken following the Rules~\ref{rule:dirty:up}~--~\ref{rule:dirty:slack}.

Since the subroutine PIVOT-DOWN$(v, (u,v))$ is called when the node $v$ wants to decrease the value of $\ell_v(u, v)$, this can never lead to a decrease in the value of $W_u$. In other words, step (01) in Figure~\ref{new:fig:pivot:down} can only  increase the weight $W_u$. Specifically, if $Y = \text{{\sc True}}$, then the weight $W_y$ increases. In contrast, if $Y = \text{{\sc False}}$, then the weight $W_u$ does not change at all. In the latter event, the node $u$ never becomes dirty, and the question of $u$ attempting to execute an ``undo'' operation does not arise. In the former event,  Rules~\ref{rule:dirty:up}~--~\ref{rule:dirty:slack} dictate that the node $u$ is about to become dirty iff $\state[u] \in \{\Up, \Downb\}$. This is the only situation where we have to check if the node $u$ can execute a successful ``undo'' operation. This situation can be split into two mutually exclusive and exhaustive cases (1) and (2), as described below. In every other situation, the node $u$ does not become dirty, it does not perform an undo operation, and the subroutine PIVOT-DOWN$(v, (u,v))$  returns the same value as $Y$. Finally, just before terminating the subroutine PIVOT-DOWN$(v, (u, v))$ we always call  UPDATE-STATUS$(u)$. The  reason for this step is explained in the beginning of Section~\ref{sub:sec:update:status}.

\paragraph{Case 1:} $Y = \text{{\sc True}}$ and $\state[u] = \Up$.

\noindent See steps (03) -- (11) in Figure~\ref{new:fig:pivot:down}. In this case, either $(u, v) \in M_{up}(u)$ or $(u, v) \notin M_{up}(u)$. In the former event, the edge $(u,v)$ has already been up-marked by $u$, and hence $u$ cannot increase the value of $\ell_u(u,v)$ any further. In the latter event, we have $\ell_u(u, v) = \ell(u)$. Before up-marking the edge $(u, v)$, the node $u$ should ensure that it satisfies Invariant~\ref{main:inv:shadow:level} (set $u = y$ and $v = x$). Hence, we must have $\ell_v(u, v) \leq \ell(u)$ if the node $u$ is to execute an undo operation. To summarise, we have to sub-cases.

\paragraph{Case 1-a:}  $(u, v) \notin M_{up}(u)$ and $\ell_v(u, v) \leq \ell(u)$. In this event,  increasing the value of $\ell_u(u, v)$ by one changes the weight $w(u,v)$ from $\beta^{-\ell(u)}$ to $\beta^{-(\ell(u)+1)}$. This undo operation is performed by calling the subroutine MOVE-UP$(u, (u,v))$.

\paragraph{Case 1-b:} Either $(u, v) \in M_{up}(u)$ and $\ell_v(u, v) > \ell(u)$. In this event, the node $u$ cannot perform an undo operation and  becomes dirty as per Rule~\ref{rule:dirty:up}.

\paragraph{Case 2:} $Y = \text{{\sc True}}$ and $\state[u] = \Downb$.

\noindent See steps (16) -- (24) in Figure~\ref{new:fig:pivot:down}. In this case, either $(u, v) \in M_{down}(u)$ or $(u, v) \notin M_{down}(u)$. In the latter event, the only way $u$ can increase the value of $\ell_u(u,v)$ is by up-marking the edge $(u, v)$. But this would result in the set $M_{up}(u)$ becoming non-empty, which in turn would violate a constraint in row (6) of Table~\ref{fig:different:states}. Hence, the node $u$ can perform an undo operation only if $(u, v) \in M_{down}(u)$. Further, if $\ell_v(u, v) \geq \ell(u)$ and $(u, v) \in M_{down}(u)$, then the weight $w(u, v)$  remains equal to $\beta^{-\ell_v(u,v)}$  even as   the value of $\ell_u(u,v)$ changes from $\ell(u) - 1$ to $\ell(u)$. This prevents  $u$ from executing an undo operation.   To summarise, there are two sub-cases.

\paragraph{Case 2-a:} We have $(u, v) \in M_{down}(u)$ and $\ell_v(u, v) < \ell(u)$. In this event, increasing the value of $\ell_u(u,v)$ by one changes the weight $w(u,v)$ from $\beta^{-(\ell(u)-1)}$ to $\beta^{-\ell(u)}$. This undo operation is performed by calling the subroutine MOVE-UP$(u, (u,v))$. 

\paragraph{Case 2-b:} Either $(u, v) \in M_{down}(u)$ or $\ell_v(u,v) \geq \ell(u)$. In this event, the node $u$ cannot perform an undo operation and becomes dirty as per Rule~\ref{rule:dirty:up}.

\begin{lemma}
\label{lm:pivot:down}
The subroutine PIVOT-DOWN$(v, (u, v))$ takes $O(\log n)$ time. It returns {\sc True} if the weight $w(u, v)$ gets changed, and {\sc False} otherwise. The node $u$ becomes dirty only if the subroutine returns {\sc True}. 
\end{lemma}

\begin{proof}
A call to the subroutine UPDATE-STATUS$(y)$ takes $O(\log n)$ time, as per Lemma~\ref{lm:time:update:status}. The rest of the proof follows from the description of the subroutine.
\end{proof}

\section{The subroutine FIX-DIRTY-NODE$(v)$}
\label{sec:fix:dirty:node}

 Note that the node $v$ undergoes a natural activation when an edge $(u,v)$ is inserted into or deleted from the graph. In contrast, the node $v$ undergoes an induced activation when some neighbour $x$ of $v$ calls the subroutine PIVOT-UP$(x, (x, v))$ or PIVOT-DOWN$(x, (x, v))$, and that subroutine returns {\sc True}.

The subroutine FIX-DIRTY-NODE$(v)$  is called immediately after the node $v$ becomes dirty due to a natural or an induced activation. Depending on the current state of  $v$, the subroutine   up-marks, down-marks or un-marks some of its incident edges $(u, v) \in E$. This involves increasing or decreasing the shadow-level $\ell_v(u, v)$ by one, for which the subroutine respectively calls PIVOT-UP$(v, (u,v))$ or PIVOT-DOWN$(v, (u,v))$. We say that a given call to PIVOT-UP$(v, (u,v))$ or PIVOT-DOWN$(v, (u,v))$ is a {\em success} if the weight $w(u,v)$ gets changed due to the call (i.e., the call returns {\sc True}), and a {\em failure} otherwise (i.e., the call returns {\sc False}). We  ensure that  one call to the subroutine FIX-DIRTY-NODE$(v)$ leads to at most one success. 

To summarise,  the subroutine  FIX-DIRTY-NODE$(v)$  makes  a series of calls to PIVOT-UP$(v, (u,v))$ or PIVOT-DOWN$(v, (u,v))$.  We terminate the subroutine  immediately after the first such call    returns {\sc True}.  We also  make the node $v$  clean just before the subroutine FIX-DIRTY-NODE$(v)$ terminates. Hence, Lemmas~\ref{lm:pivot:up},~\ref{lm:pivot:down} imply the following observation.

\begin{observation}
\label{ob:fix:dirty:node}
The node $v$ becomes clean at the end of the subroutine FIX-DIRTY-NODE$(v)$. Furthermore, during a call to the subroutine FIX-DIRTY-NODE$(v)$, at most one neighbour of the node $v$ becomes dirty. 
\end{observation}

\begin{figure}[htbp]
\centerline{\framebox{
\begin{minipage}{5.5in}
\begin{tabbing}
01. \  \=   {\sc If} $\text{{\sc State}}[v] = \text{{\sc Up}}$ \\
02. \> \qquad \= $\text{FIX-UP}(v)$ \\ 
03. \> {\sc Else if} $\text{{\sc State}}[v] = \text{{\sc Down-B}}$ \\
04. \> \> $\text{FIX-DOWN-B}(v)$\\ 
05. \> {\sc Else if} $\text{{\sc State}}[v] = \text{{\sc Down}}$ and $\ell(y) > K$ \\
06. \> \> $\text{FIX-DOWN}(v)$ \\ 
07. \> {\sc Else if} $\text{{\sc State}}[v] = \text{{\sc Up-B}}$ \\
08. \> \> $\text{FIX-UP-B}(v)$  \\   
09. \> UPDATE-STATUS($v$)
\end{tabbing}
\end{minipage}
}}
\caption{\label{new:fig:fix:dirty} FIX-DIRTY-NODE($v$).}
\end{figure}

We now describe the subroutine FIX-DIRTY-NODE$(v)$ in a bit more detail. See Figure~\ref{new:fig:fix:dirty}. Note that the node $v$ becomes dirty only if it experiences an activation, and the subroutine FIX-DIRTY-NODE$(v)$ is called immediately after the node $v$ becomes dirty.  Thus, Rule~\ref{rule:dirty:slack} and Corollary~\ref{cor:activation} imply that at the beginning of the subroutine FIX-DIRTY-NODE$(v)$ we must have: either (1) $\state[v] = \Up$, or (2) $\state[v] = \Downb$, or (3) $\state[v] = \Down$ and $\ell(y) > K$, or (4) $\state[v] = \Upb$. Accordingly,  we   call one of the four  subroutines: FIX-UP$(v)$, FIX-DOWN-B$(v)$, FIX-DOWN$(v)$ and FIX-UP-B$(v)$. For the rest of  Section~\ref{sec:fix:dirty:node}, we focus on  describing  these four  subroutines. Note that we call  UPDATE-STATUS$(v)$ just before terminating the subroutine FIX-DIRTY-NODE$(v)$, for a reason that is explained in the beginning of  Section~\ref{sub:sec:update:status}. We now give a bound on the runtime of the subroutine, which follows from Lemmas~\ref{lm:fix:up},~\ref{lm:fix:down:b},~\ref{lm:fix:down},~\ref{lm:fix:up:b} and~\ref{lm:time:update:status}.

\begin{lemma}
\label{lm:fix:dirty:node}
The subroutine FIX-DIRTY-NODE$(v)$ takes  $O(\log^2 n)$ time.
\end{lemma}

\begin{figure}[htbp]
\centerline{\framebox{
\begin{minipage}{5.5in}
\begin{tabbing}
01. \  \= $D[v] \leftarrow 0$,  $i \leftarrow \ell(v)$ \\
02. \>   Pick an edge $(u, v) \in E_i(v)$. \\ 
03. \>  $\text{PIVOT-UP}(v, (u, v))$  
\end{tabbing}
\end{minipage}
}}
\caption{\label{new:fig:fix:up} FIX-UP($v$).}
\end{figure}

\subsection{FIX-UP$(v)$.} See Figure~\ref{new:fig:fix:up}. This subroutine is called when a node $v$  with $\state[v] = \Up$ becomes dirty due to an activation. This activation must have increased the weight $W_v$. See Rule~\ref{rule:dirty:up} and Case (1) of its subsequent justification. Let $i = \ell(v)$ be the current level of the node. Since $\state[v] = \Up$, we must have $E_i(v) \neq \emptyset$ as per row (1) of Table~\ref{fig:different:states}. The node $v$ picks any edge $(u, v) \in E_i(v)$ and up-marks that edge by calling the subroutine PIVOT-UP$(v, (u,v))$. See the justification for Rule~\ref{rule:dirty:up}. Since $(u, v) \in E_i(v)$ just before this step, we must have $\ell_u(u, v) \leq i$. This means that increasing the shadow-level $\ell_v(u, v)$  from $i$ to $(i+1)$ changes the weight $w(u, v)$ from $\beta^{-i}$ to $\beta^{-(i+1)}$. In other words, the very first call to PIVOT-UP$(v, (u,v))$ becomes a {\em success}. Thus, we  terminate the subroutine. Lemma~\ref{lm:fix:up} now follows from Lemma~\ref{lm:pivot:up}.

\begin{lemma}
\label{lm:fix:up}
The runtime of  FIX-UP$(v)$ is $O(\log n)$.
\end{lemma}

\begin{figure}[htbp]
\centerline{\framebox{
\begin{minipage}{5.5in}
\begin{tabbing}
01. \ \= $D[v] \leftarrow 0$,   $k \leftarrow 0$ \\
02. \> {\sc While} $k < \beta^{5}$ \\
03. \>  \qquad \= $k \leftarrow k+1$ \\
04. \> \> {\sc If} $M_{down}(v) = \emptyset$ \\
05. \> \> \qquad BREAK \\
06. \> \>  Pick an edge $(u, v) \in M_{down}(v)$. \\ 
07. \> \> $X \leftarrow \text{PIVOT-UP}(v, (u, v))$ \\
08. \> \> {\sc If} $X = \text{{\sc True}}$ \\
09. \> \> \qquad  BREAK 
\end{tabbing}
\end{minipage}
}}
\caption{\label{new:fig:fix:down:b} FIX-DOWN-B($v$).}
\end{figure}

\subsection{FIX-DOWN-B$(v)$.} See Figure~\ref{new:fig:fix:down:b}.  This subroutine is called when a node $v$  with $\state[v] = \Downb$ becomes dirty due to an activation. This activation must have increased the weight $W_v$. See Rule~\ref{rule:dirty:up} and Case (2) of its subsequent justification.   Since $\state[v] = \Downb$, we must have $M_{down}(v) \neq \emptyset$ as per row (6) of Table~\ref{fig:different:states}. 

\smallskip
\noindent {\em The node $v$ picks an edge $(u, v) \in M_{down}(v)$, and un-marks it by calling  PIVOT-UP$(v, (u, v))$.}

\smallskip
\noindent
We keep repeating the above step until one of three events occurs: (1) The set $M_{down}(v)$ becomes empty. (2) We make the $\beta^{5}$-th call to PIVOT-UP$(v, (u, v))$.  (3) We encounter the first call to PIVOT-UP$(v, (u, v))$ which leads to a change in the weight $w(u, v)$. We then terminate the subroutine. By Lemma~\ref{lm:pivot:up},  each iteration of the {\sc While} loop in Figure~\ref{new:fig:fix:down:b} takes $O(\log n)$ time. This gives us the following lemma.

\begin{lemma}
\label{lm:fix:down:b}
The subroutine  FIX-DOWN-B$(v)$ takes $O(\beta^{5}  \log n) = O(\log n)$ time, for constant $\beta$.
\end{lemma}

\begin{figure}[htbp]
\centerline{\framebox{
\begin{minipage}{5.5in}
\begin{tabbing}
01. \ \= $D[v] \leftarrow 0$,  $i \leftarrow \ell(v)$,  $k \leftarrow 0$ \\
02. \> {\sc While} $k < \beta^{5} L$ \\
03. \>  \qquad \= $k \leftarrow k+1$ \\
04. \> \> {\sc If} $E_i(v) \setminus M_{down}(v) = \emptyset$ \\
05. \> \> \qquad  BREAK \\
06. \> \>  Pick an edge $(u, v) \in E_{i}(v) \setminus M_{down}(v)$. \\
07. \> \> $X \leftarrow \text{PIVOT-DOWN}(v, (u, v))$ \\
08. \> \> {\sc If} $X = \text{{\sc True}}$ \\
09. \> \> \qquad  BREAK 
\end{tabbing}
\end{minipage}
}}
\caption{\label{new:fig:fix:down} FIX-DOWN($v$).}
\end{figure}

\subsection{FIX-DOWN$(v)$.} See Figure~\ref{new:fig:fix:down}.
This subroutine is called when a node $v$  with $\state[v] = \Down$ and $\ell(v) > K$ becomes dirty due to an activation. This activation must have decreased the weight $W_v$. See Rule~\ref{rule:dirty:down} and Case (1) of its subsequent justification.  Let $i = \ell(v)$ be the current level of the node $v$. Since $\state[v] = \Down$ and $\ell(y) > K$, we must have $E_i(v) \setminus M_{down}(v) \neq \emptyset$ as per row (2) of Table~\ref{fig:different:states}. 

\smallskip
\noindent {\em The node  $v$ picks an edge $(u, v) \in E_i(v) \setminus  M_{down}(v)$, and down-marks it by calling  PIVOT-DOWN$(v, (u, v))$.}

\smallskip
\noindent
We keep repeating the above step until one of three events occurs: (1) The set $E_i(v) \setminus M_{down}(v)$ becomes empty. (2) We make the $\beta^{5} L$-th call to PIVOT-DOWN$(v, (u, v))$.  (3) We encounter the first call to PIVOT-DOWN$(v, (u, v))$ which leads to a change in the weight $w(u, v)$. We then terminate the subroutine.

We now explain how to select an edge from  $E_i(v) \setminus M_{down}(v)$ in step (06) of Figure~\ref{new:fig:fix:down}. Recall that we maintain  the sets $E_i(v)$ and $M_{down}(v)$ as balanced search trees as per Section~\ref{sub:sec:data:structures}. Specifically, we maintain the elements of $E_i(v)$  in a {\em particular order}. This ordered list  is partitioned into two disjoint blocks: The first block consists of the edges in $E_i(v)\setminus M_{down}(v)$, and the second block consists of the edges in $E_i(v) \cap M_{down}(v)$. During a given iteration of the {\sc While} loop in Figure~\ref{new:fig:fix:down}, we pick an edge $(u, v)$ that comes first  in this ordering of $E_i(v)$ and check if $(u, v) \in M_{down}(v)$. If yes, then we know for sure that $E_i(v) \setminus M_{down}(v) = \emptyset$, and hence we terminate the subroutine. Else if $(u, v) \notin M_{down}(v)$, then $v$ down-marks the edge by calling PIVOT-DOWN$(v, (u,v))$. Now, consider two cases.

\begin{enumerate}
\item The call to PIVOT-DOWN$(v, (u,v))$ is a {\em failure}. It means that the weight and the level of the edge $(u,v)$ do not change during the call. Hence, at the end of the call we get: $(u,v) \in E_i(v)$ and $(u, v) \in M_{down}(v)$. At this point we delete the edge $(u,v)$ from $E_i(v)$. Immediately afterward we again insert the edge $(u, v)$ back  to  $E_i(v)$, but this time $(u, v)$ occupies the last position in the ordering of $E_i(v)$. Hence, the ordering of $E_i(v)$ remains correctly partitioned into two blocks as described above. 

\item The call to PIVOT-DOWN$(v, (u, v))$ is a {\em success}. It means that the weight and the level of the edge $(u,v)$ changes during the call. At the end of the call we get: $(u, v) \notin E_i(v)$ and $(u, v) \in M_{down}(v)$. At this point we terminate the subroutine FIX-DOWN$(v)$. By Lemma~\ref{lm:pivot:down},  an iteration of the {\sc While} loop in Figure~\ref{new:fig:fix:down} takes $O(\log n)$ time.  This implies the lemma below.
\end{enumerate}

\begin{lemma}
\label{lm:fix:down}
The subroutine  FIX-DOWN$(v)$ takes $O(\beta^{5} L \cdot \log n) = O(\log^2 n)$ time, for constant $\beta$.
\end{lemma}

\begin{figure}[htbp]
\centerline{\framebox{
\begin{minipage}{5.5in}
\begin{tabbing}
01. \  \= $D[v] \leftarrow 0$, $i \leftarrow \ell(v)$, $k \leftarrow 0$ \\
02. \> {\sc While} $k < \beta^{5}$ \\
03. \>  \ \ \ \ \ \ \ \ \= $k \leftarrow k+1$ \\
04. \> \> {\sc If} $M_{up}(v) = \emptyset$ \\
05. \> \> \qquad  BREAK \\
06. \> \>  Pick an edge $(u, v) \in M_{up}(v)$. \\ 
07. \> \> $X \leftarrow \text{PIVOT-DOWN}(v, (u, v))$ \\
08. \> \> If $X = \text{{\sc True}}$ \\ 
09. \> \> \qquad  BREAK 
\end{tabbing}
\end{minipage}
}}
\caption{\label{new:fig:fix:up:b} FIX-UP-B($v$).}
\end{figure}

\subsection{FIX-UP-B$(v)$.} See Figure~\ref{new:fig:fix:up:b}.
This subroutine is called when a node $v$  with $\state[v] = \Upb$ becomes dirty due to an activation. This activation must have decreased the weight $W_v$. See Rule~\ref{rule:dirty:down} and Case (2) of its subsequent justification.   Since $\state[v] = \Upb$, we must have $M_{up}(v) \neq \emptyset$ as per row (5) of Table~\ref{fig:different:states}. 

\smallskip
\noindent {\em The node $v$ picks an edge $(u, v) \in M_{up}(v)$, and un-marks it by calling  PIVOT-DOWN$(v, (u, v))$.}

\smallskip
\noindent
We keep repeating the above step until one of three events occurs: (1) The set $M_{up}(v)$ becomes empty. (2) We make the $\beta^{5}$-th call to PIVOT-DOWN$(v, (u, v))$.  (3) We encounter the first call to PIVOT-DOWN$(v, (u, v))$ which leads to a change in the weight $w(u, v)$. We then terminate the subroutine. By Lemma~\ref{lm:pivot:down},  each iteration of the {\sc While} loop in Figure~\ref{new:fig:fix:up:b} takes $O(\log n)$ time.  This gives us the following lemma.

\begin{lemma}
\label{lm:fix:up:b}
The subroutine  FIX-UP-B$(v)$ takes $O(\beta^{5} \cdot \log n) = O(\log n)$ time, for constant $\beta$.
\end{lemma}

\section{Handling the insertion or deletion of an edge}
\label{sec:insert:delete}

In this section, we explain how our algorithm  handles the insertion/deletion of an edge in the input graph.

\paragraph{Insertion of an edge $(u, v)$.} 
We set $\ell_u(u, v) \leftarrow \ell(u)$, $\ell_v(u, v) \leftarrow \ell(v)$ and $\ell(u, v) \leftarrow \max(\ell_u(u,v), \ell_v(u,v))$. The newly inserted edge gets a weight $w(u, v) \leftarrow \beta^{-\ell(u, v)}$. Hence, each of the node-weights $W_u$ and $W_v$ also increases by $\beta^{-\ell(u, v)}$. This amounts to a natural activation for each of the endpoints $\{u, v\}$.  For every endpoint $x \in \{u, v\}$, we now decide if $x$  should become dirty due to this activation. This decision is taken as per Rules~\ref{rule:dirty:up}~--~\ref{rule:dirty:slack}. We now call the subroutine UPDATE-STATUS$(x)$ for $x \in \{u, v\}$, for reasons explained in the beginning of Section~\ref{sub:sec:update:status}. Finally, we call the subroutine FIX-DIRTY($.$) as described in Figure~\ref{new:fig:dirty}.

\paragraph{Deletion of an edge $(u,v)$.}
 Just before the edge-deletion, its weight was  $w(u, v)$. We first decrease each of the node-weights $W_u, W_v$ by $w(u,v)$. Then we delete all the data structures associated with the edge $(u,v)$. This amounts to a natural activation for each of its endpoints. For every node $x \in \{u, v\}$, we  decide if $x$  should become dirty due to this activation, as per Rules~\ref{rule:dirty:up}~--~\ref{rule:dirty:slack}. At this point, we  call the subroutine UPDATE-STATUS$(x)$ for $x \in \{u, v\}$, for reasons explained in the beginning of Section~\ref{sub:sec:update:status}.  Finally, we call the subroutine FIX-DIRTY($.$) as per Figure~\ref{new:fig:dirty}. 

\begin{figure}[htbp]
\centerline{\framebox{
\begin{minipage}{5.5in}
\begin{tabbing}
   {\sc While} there exists a dirty node $x \in V$: \\
  \qquad  \= FIX-DIRTY-NODE($x$)  \qquad   // See Section~\ref{sec:fix:dirty:node}.
\end{tabbing}
\end{minipage}
}}
\caption{\label{new:fig:dirty} FIX-DIRTY($.$).}
\end{figure}

\paragraph{Two assumptions.} For ease of analysis, we will make two simplifying assumptions. At first glance, these assumptions might seem highly restrictive. But we will explain how the analysis can be extended to the general setting, where these assumptions need not hold, by slightly modifying our algorithm.

\begin{assumption}
\label{assume:dirty}
The insertion or deletion of an edge $(u,v)$ makes at most one of its endpoints dirty. 
\end{assumption}

\noindent {\em Justification.} Consider a scenario where  the insertion or deletion of an edge $(u, v)$ is about to make both its endpoints dirty. Without any loss of generality, suppose that the weight of $v$  increases by $\delta_v$ due to this edge insertion or deletion. Note that $\delta_v$ can also be negative. We  reset the weight $W_v$ to the value it had just before the edge insertion or deletion took place, by setting $W_v \leftarrow W_v - \delta_v$.  In other words,  the node $v$ becomes {\em blind} to the fact that its weight has changed. Clearly, after this simple modification, only the node $u$  becomes dirty. We now go ahead and call the subroutine FIX-DIRTY(.). Starting from the node $u$, this creates a {\em chain} of calls to  FIX-DIRTY-NODE($x$) for different $x \in V$ (see Observation~\ref{ob:fix:dirty:node}). When this chain stops, we go back and update the weight of the other endpoint $v$, by setting $W_v \leftarrow W_v + \delta_v$. So the node $v$ now {\em wakes up} and experiences an activation. If the node $v$ becomes dirty due to this activation, as per Rules~\ref{rule:dirty:up}~--~\ref{rule:dirty:slack}, then we again go ahead and call the subroutine FIX-DIRTY(.). Starting from the node $v$, this creates a second chain of calls to the subroutine FIX-DIRTY-NODE($x$) for different $x \in V$. When this second chain stops, we conclude that we have successively handled the insertion or deletion of the edge $(u, v)$.

\begin{assumption}
\label{assume:weight}
The weight $W_u$ of a node $u$ changes by at most $\beta^{-(\ell(u)+1)}$ due to a natural activation.
\end{assumption}

\noindent {\em Justification.} For any edge $(u, v)$, we have: $\ell(u, v) \geq \ell_u(u, v) \geq \ell(u) -1$, and $w(u, v) = \beta^{-\ell(u,v)}$. So the weight of any edge incident on $u$ is at most $\Delta_u = \beta^{-(\ell(u)-1)}$. Now, suppose that Assumption~\ref{assume:weight} gets violated. Specifically, the weight $W_u$ changes by $\Delta'_u$ due to a natural activation, where $\Delta'_u > \beta^{-(\ell(u)+1)}$. To handle this situation, we fix the node $u$ in $r_u$ {\em rounds}, where $r_u = \Delta'_u/\beta^{-(\ell(u)+1)} \leq \Delta_u/\beta^{-(\ell(u)+1)} \leq \beta^2$. In each round, we change the weight $W_u$ by  $\beta^{-(\ell(u)+1)}$ and call the subroutine FIX-DIRTY(.). This way the node becomes {\em oblivious} to the fact that Assumption~\ref{assume:weight} gets violated. The update time increases by a factor of $r_u$, which is $O(1)$ for constant $\beta$.

\paragraph{Analysis of our algorithm.}
Just before the insertion or deletion of  the edge $(u,v)$, every node in the graph is clean, and every node satisfies the constraints corresponding to its current state as specified by Table~\ref{fig:different:states}. By Assumption~\ref{assume:dirty}, at most one endpoint $x \in \{ u, v \}$ becomes dirty due to this edge insertion/deletion. Hence, at most one node is dirty in the beginning of the call to the subroutine FIX-DIRTY(.). By Observation~\ref{ob:fix:dirty:node}, we get a {\em chain} of calls to the subroutine FIX-DIRTY-NODE$(y)$ for $y \in V$. Each call to FIX-DIRTY-NODE$(y)$ makes at most one neighbour of $y$ dirty, which is fixed at the next iteration of the {\sc While} loop in Figure~\ref{new:fig:dirty}. Thus, at every point in time there is at most one dirty node in the entire graph. We now prove two theorems.

\begin{theorem}
\label{th:no:failure}
While handling a sequence of edges insertions and deletions, our algorithm never HALTS due  to a call to the subroutine UPDATE-STATUS$(y)$.
\end{theorem}

The proof of Theorems~\ref{th:no:failure}  appears in Section~\ref{sec:th:no:failure}. Recall the discussion in the first paragraph of Section~\ref{sub:sec:update:status}. To summarise that discussion, Theorem~\ref{th:no:failure} ensures that throughout the duration of our algorithm, every node satisfies the constraints corresponding to its current state as per Table~\ref{fig:different:states}. Hence, by Lemma~\ref{lm:different:states}, conditions (2) and (3) of Definition~\ref{def:structure} continue to remain satisfied all the time. This observation, along with Corollary~\ref{cor:inv:shadow:level}, implies that our algorithm successfully maintains a nice-partition as per Definition~\ref{def:structure}. 

In Theorem~\ref{th:update:time}, we  bound the worst-case update time of our algorithm. The proof of this theorem appears in Section~\ref{new:sec:th:update:time}. Intuitively, we  show that after four consecutive calls to FIX-DIRTY-NODE$(x)$ in the {\sc While} loop of Figure~\ref{new:fig:dirty}, the value of $\ell(x)$ decreases by at least one. Since $\ell(x) \in [K,L]$ for every node $x \in V$, there can be at most $4(L-K+1) = O(\log n)$ iterations of the {\sc While} loop of Figure~\ref{new:fig:dirty}. By Lemma~\ref{lm:fix:dirty:node}, each iteration of this {\sc While} loop takes $O(\log^2 n)$ time. Accordingly, the subroutine FIX-DIRTY(.) as described in Figure~\ref{new:fig:dirty} takes $O(\log^3 n)$ time, and this  gives an upper bound on the worst-case update time of our algorithm.

\begin{theorem}
\label{th:update:time}
Our algorithm handles an edge insertion or deletion  in $O(\log^3 n)$ worst-case time.
\end{theorem}

The main result of this paper (Theorem~\ref{th:main})  follows from Theorems~\ref{th:no:failure} and~\ref{th:update:time}.

\subsection{Recap of our algorithm.} 
\label{sub:sec:recap}
During the course of our algorithm, the weight of a node $x$ can change only under three scenarios: 
\begin{itemize}
\item (1) An edge  incident to $x$ gets inserted or deleted. 
\item (2) A neighbour $y$ of $x$ makes a call to PIVOT-UP$(y, (x, y))$ or PIVOT-DOWN$(y, (x, y))$ and the call returns {\sc True}.  In this scenario, the weight-change occurs only during the call to MOVE-UP$(y, (x,y))$ or MOVE-DOWN$(y, (x, y))$. 
\item (3) The node $x$ makes a call to FIX-DIRTY-NODE$(x)$. 
\end{itemize}
\noindent Note that scenarios (2) and (3) are {\em symmetric}: a call is made to the subroutine PIVOT-UP$(y, (x, y))$ or PIVOT-DOWN$(y, (x, y))$ only when  $y$ itself is executing FIX-DIRTY-NODE$(y)$. Scenarios (1) and (2) respectively correspond to a natural and an induced activation of $x$. A node $x$ can become dirty only due to a natural or an induced activation, as per Rules~\ref{rule:dirty:up}--~\ref{rule:dirty:slack}. Scenario (3) is the {\em response} of  $x$ after it becomes dirty. At the end of the call to FIX-DIRTY-NODE$(x)$ in scenario (3), the node $x$ becomes clean again. 

During the course of our algorithm, the shadow-level $\ell_y(x, y)$ of an edge $(x, y)$ increases iff a call is made to  MOVE-UP$(y, (x,y))$, and decreases iff a call is made to  MOVE-DOWN$(y, (x,y))$.     These two subroutines are defines in Sections~\ref{sub:sec:pivot:up} and~\ref{sub:sec:pivot:down}. A call to MOVE-DOWN$(y, (x,y))$ is made only if we are executing the subroutine PIVOT-DOWN$(y, (x, y))$. In contrast, a call to MOVE-UP$(y, (x,y))$ is made only if we are executing either the subroutine PIVOT-UP$(y, (x, y))$ or the subroutine PIVOT-DOWN$(x,(x, y))$.

\section{Proof of Theorem~\ref{th:no:failure}}
\label{sec:th:no:failure}
Let $\mathcal{U} = \{ \Up, \Down, \Slack, \Idle, \Upb, \Downb \}$ be the set of all possible states of a node (see Table~\ref{fig:different:states}). For the rest of this section, we assume that our algorithm HALTS at a time-instant  (say) $t_1$ due to a call made to   UPDATE-STATUS$(x)$ for some node  $x \in V$. Suppose that $\state[x] = S$ at time $t_1$. To prove Theorem~\ref{th:no:failure},  it suffices to derive a contradiction for all $S \in \mathcal{U}$. These contradictions are derived in Sections~\ref{new:sec:first}~--~\ref{new:sec:last}. 

 Let $t_0 < t_1$  be the unique time-instant such that: (1)  $\state[x] = S$ throughout the time-interval $[t_0, t_1]$ and  (2) $\state[x] \neq S$ just before time-instant $t_0$.  During the time-interval $[t_0, t_1]$, the node-weight $W_x$ can change due to three types of events: We  classify these types as  A,  B and  C, and specify each of them below.

\begin{itemize}
\item {\em Type A:  An activation of $x$ increases the weight $W_x$.}

\item {\em Type B:  An activation of $x$ decreases the weight $W_x$.}

\item {\em Type C: We call the subroutine FIX-DIRTY-NODE$(x)$.}
\end{itemize}

\smallskip
\noindent  Rules~\ref{rule:dirty:up}~--~\ref{rule:dirty:slack} dictate whether or not the node $x$  becomes dirty after an event of Type A or B.
 A Type C event occurs  when $x$ becomes dirty due to  a Type A or Type B event.  The node $x$ becomes clean again before the call to FIX-DIRTY-NODE$(x)$ ends.

\subsection{Deriving a contradiction for $S = \Upb$}
\label{new:sec:first}
The only way the node $x$ can change its level is if we execute the steps in Case 2-a or 2-b during a call to UPDATE-STATUS$(x)$. This situation can never occur during the time-interval $[t_0, t_1]$, throughout which we have $\state[x] = S = \Upb$. Thus, the node $x$ stays at the same level throughout the time-interval $[t_0, t_1]$. 

In Claims~\ref{new:cl:halt:upb:1} and~\ref{new:cl:halt:upb:2}, we respectively bound the weight $W_x$ at time-instants $t_0$ and $t_1$. In Corollary~\ref{cor:halt:upb}, we use these two claims to bound the change in the weight $W_x$ during the time-interval $[t_0, t_1]$.

\begin{claim}
\label{new:cl:halt:upb:1}
 $W_x \geq 1 - 1/\beta - 1/\beta^K$ at time $t_0$.
\end{claim}

\begin{proof}
The node $x$ undergoes an activation at time $t_0$ which changes its state to $\Upb$. As per the discussion in Section~\ref{sub:sec:dirty:node}, the only way this can happen is if $\state[x] = \Up$ just before time $t_0$ and the activation at time $t_0$ decreases  $W_x$ (see Case 1 in the justification for Rule~\ref{rule:dirty:up}). Thus, 
row (1) in Table~\ref{fig:different:states} gives us: $W_x \geq 1 - 1/\beta$ just before time $t_0$.   Since the weight of an edge is at most $\beta^{-K}$, the activation of   $x$ at time $t_0$ changes  $W_x$ by at most $\beta^{-K}$. So  we get: $W_x \geq 1 - 1/\beta - 1/\beta^K$  after the activation of $x$ at time $t_0$.
\end{proof}

\begin{claim}
\label{new:cl:halt:upb:2}
$W_x < 1 - 2/\beta$ and $M_{up}(x) \neq \emptyset$ at time $t_1$.
\end{claim}

\begin{proof}
$\state[x] = \Upb$ throughout the time-interval $[t_0, t_1]$. Just before time $t_1$, the node $x$ undergoes an activation, say, {\bf a*}. Subsequent to the activation {\bf a*}, our algorithm HALTS during a call to  UPDATE-STATUS$(x)$ at time $t_1$. From row (5) of Table~\ref{fig:different:states}, we get: $M_{down}(x) = \emptyset$, $M_{up}(x) \neq \emptyset$ and $1-2/\beta \leq W_x < 1-1/\beta$ just before the activation {\bf a*}. No edge gets inserted into the sets $M_{up}(x)$ and $M_{down}(x)$ during the activation {\bf a*}. Since the algorithm HALTS at time $t_1$, the activation {\bf a*} must have changed the weight $W_x$ in such a way that the node $x$ violates the constraints for every state as defined in Table~\ref{fig:different:states}. This can happen only if $W_x < 1 - 2/\beta$ and $M_{up}(x) \neq \emptyset$ at time $t_1$.
\end{proof}

\begin{corollary}
\label{cor:halt:upb}
During the interval $[t_0, t_1]$, the node-weight $W_x$  decreases by at least $1/\beta - 1/\beta^K$.
\end{corollary}

\begin{proof}
Follows from Claims~\ref{new:cl:halt:upb:1} and~\ref{new:cl:halt:upb:2}.
\end{proof}

\begin{claim}
\label{new:cl:halt:upb:dirty}
An event of Type A does not make the node $x$ dirty. An event of Type B makes the node $x$ dirty.
\end{claim}

\begin{proof}
Throughout the time-interval $[t_0, t_1]$, we have $\state[x] = \Upb$. So the claim follows from Rule~\ref{rule:dirty:down}.
\end{proof}

In the next three claims, we bound the change in  $W_x$ that can result from an event of Type B or C. 

\begin{claim}
\label{new:cl:halt:upb:activation:B}
The node-weight $W_x$ decreases by at most $\Delta = \beta^{-\ell(x)} - \beta^{-\ell(x)-1}$ due to a Type B event.
\end{claim}

\begin{proof}
If the Type B event occurs due a natural activation of $x$, then the claim follows from Assumption~\ref{assume:weight} since $\beta^{-\ell(x) - 1} \leq \beta^{-\ell(x)} - \beta^{-\ell(x)-1}$ as long as $\beta \geq 2$. For the rest of the proof, suppose that the Type B event occurs due to an induced activation. This means that the Type B event results from some neighbour $y$ of $x$ increasing the value of $\ell_y(x, y)$ from, say, $i$ to $(i+1)$. For this to change the weight $w(u,v)$, we must have $\ell_x(x, y) \leq i$. Since $\state[x] = \Upb$, row (5) of Table~\ref{fig:different:states} implies that $M_{down}(x) = \emptyset$ and hence $\ell(x) \leq \ell_x(x, y) \leq i$.  
It follows that the weight  $W_x$ decreases by  $\beta^{-i} - \beta^{-(i+1)} \leq \beta^{-\ell(x)}  - \beta^{-(\ell(x)+1)}$.  
\end{proof}

Consider an event of Type C. This event occurs when we call the subroutine FIX-DIRTY-NODE$(x)$. Since $\state[x] = \Upb$, this in turn leads to a call to the subroutine FIX-UP-B$(x)$. See Figures~\ref{new:fig:fix:dirty} and~\ref{new:fig:fix:up:b}. Hence, during a Type C event, the node $x$ un-marks one or more incident edges $(x, y) \in M_{up}(x)$ by calling the subroutine PIVOT-DOWN$(x, (x, y))$. If the un-marking of an edge $(x, y)$ changes its weight $w(x, y)$, then we say that the un-marking is a {\em success}; otherwise the un-marking is a {\em failure}. Figure~\ref{new:fig:fix:dirty} ensures that an event of Type C leads to at most one success.

\begin{claim}
\label{new:cl:halt:upb:activation:C}
If a Type C event leads to a success, then it increases the node-weight $W_x$ by $\Delta = \beta^{-\ell(x)} - \beta^{-\ell(x) - 1}$.
\end{claim}

\begin{proof}
Let the success correspond to the un-marking of the edge $(x, y)$.
 Just before this un-marking, we have $(x, y) \in M_{up}(x)$ and hence $\ell_x(x, y) = \ell(x) + 1$. The un-marking reduces the value of $\ell_x(x, y)$ from $\ell(x) + 1$ to $\ell(x)$. For this to change the weight $w(x, y)$, the value of $\ell(x, y)$ must also have decreased from $\ell(x)+1$ to $\ell(x)$ due to the un-marking. This means that the weight $w(x, y)$ increases by an amount $\Delta = \beta^{-\ell(x)} - \beta^{-\ell(x) - 1}$ due to the un-marking. Since any Type C event leads to at most one success,  the weight $W_x$ also changes by exactly $\Delta$ during the Type C event under consideration.
 \end{proof}

\begin{claim}
\label{new:cl:halt:upb:activation:C:failure}
If a Type C event does not lead to a success, then it does not change the weight $W_x$, and $\beta^5$ edges get deleted from the set $M_{up}(x)$ due to such a Type C event. 
\end{claim}

\begin{proof}
Consider a Type C event  that does not lead to a success. During this event, each time the node $x$ un-marks an edge $(x, y)$, it leads to a failure and does not change the weight $w(x, y)$. Thus, the weight $W_x$ also does not change due to such an event of Type C.

Suppose that the Type C event under consideration  leads to zero success and less than $\beta^5$ failures. This  implies that the subroutine FIX-UP-B$(x)$ terminates due to step (05) in Figure~\ref{new:fig:fix:up:b}, and thus $M_{up}(x) = \emptyset$ at this point in time. Next,  the subroutine FIX-DIRTY-NODE$(x)$    calls UPDATE-STATUS$(x)$  as per step (09) in Figure~\ref{new:fig:fix:dirty}, which in turn changes the state of the node $x$ since we cannot simultaneously have $\state[x] = \Upb$ and $M_{up}(x) = \emptyset$. See row (5) of Table~\ref{fig:different:states}. However, this leads us to a contradiction, for we have assumed that $\state[x] = \Upb$ throughout the time-interval $[t_0, t_1]$.
\end{proof}

Let $n_B$ and $n_C$ respectively denote the number of  Type B and Type C events during the time-interval $[t_0, t_1]$. Let $n_C^s$ (resp. $n_C^f$) denote the number of Type C events during the time-interval $[t_0, t_1]$ that lead  (resp. do not lead) to a success. Clearly, we have: $n_C = n_C^s + n_C^f$. 
By Claim~\ref{new:cl:halt:upb:dirty}, every Type B event is followed by a Type C event. Hence, we get: $n_B \leq n_C = n_C^s + n_C^f$, which implies that:
\begin{equation}
\label{eq:F:2}
n_C^f \geq n_B  - n_C^s 
\end{equation}
Any change in the weight $W_x$ during the time-interval $[t_0, t_1]$ results  from an event of Type A, B or C.  Now,  an event of Type A  increases the weight $W_x$, an event of Type B decreases the weight $W_x$ by at most $\Delta$ (see Claim~\ref{new:cl:halt:upb:activation:B}),  an event of Type C that leads to a success increases the weight $W_x$ by $\Delta$ (see Claim~\ref{new:cl:halt:upb:activation:C}), and an event of Type C that does not lead to a success leaves the value of $W_x$ unchanged (see Claim~\ref{new:cl:halt:upb:activation:C:failure}). Since the weight $W_x$ decreases by at least $1/\beta - 1/\beta^K$ during the time-interval $[t_0, t_1]$ (see Corollary~\ref{cor:halt:upb}), we get:
\begin{equation}
\label{eq:F:3}
(n_B - n_C^s) \cdot \Delta \geq 1/\beta - 1/\beta^K
\end{equation}
 Claim~\ref{new:cl:halt:upb:activation:B} gives: $1/\Delta \geq \beta^{\ell(x)}$. By eq~(\ref{eq:beta:K:L}), we have $1/\beta - 1/\beta^K \geq 1/\beta^2$. Thus, eq~(\ref{eq:F:2}) and~(\ref{eq:F:3}) give: \begin{equation}
\label{eq:F:4}
n_C^f \geq (1/\Delta) \cdot (1/\beta - 1/\beta^K) \geq \beta^{\ell(x)-2} 
\end{equation}
By Claim~\ref{new:cl:halt:upb:activation:C:failure}, for each Type C event that contributs to $n_C^f$, the node $x$ deletes $\beta^5$ edges from  $M_{up}(x)$. Hence, eq.~(\ref{eq:F:4})  implies that during the time-interval $[t_0, t_1]$, the node $x$ deletes $n_C^f \cdot \beta^5 \geq \beta^{\ell(x)+3}$ edges from $M_{up}(x)$.  Furthermore, the node $x$ never inserts an edge into the set $M_{up}(x)$ during the time-interval $[t_0, t_1]$, for  $\state[x] = \Upb$ throughout this time-interval (see Figure~\ref{new:fig:fix:up:b} and Section~\ref{sub:sec:recap}). Thus, we have:
\begin{equation}
\label{eq:F:5}
|M_{up}(x)| \geq \beta^{\ell(x)+3} \text{ at time-instant } t_0.
\end{equation} 
Note that every edge $(x, v) \in M_{up}(x)$ has $\ell_x(x, v) = \ell(x) + 1$ and $\ell_v(x, v) \leq \ell(x)$ by Invariant~\ref{main:inv:shadow:level}. Thus, the weight of every edge $(x, v) \in M_{up}(x)$ is given by $w(x, y) = \beta^{-\ell(x)-1}$. By equation~\ref{eq:F:5}, we now derive that $W_x \geq \sum_{(x, v) \in M_{up}(x)} w(x,v) \geq |M_{up}(x)| \cdot \beta^{-\ell(x) - 1} > 1$ at time-instant $t_0$. This leads to a contradiction, since $\state[x] = \Upb$ at time $t_0$ and hence row(5) of Table~\ref{fig:different:states} requires that $W_x < 1-1/\beta$.

\subsection{Deriving a contradiction for $S = \Up$.}

Just before time $t_1$, the node $x$ has $\state[x] = \Up$, $M_{down}(x) = \emptyset$, $E_{\ell(x)}(x) \neq \emptyset$, and $1-1/\beta \leq W_x < 1$. See row (1) of Table~\ref{fig:different:states}. Consider the activation of $x$ at time $t_1$ that results in a call to UPDATE-STATUS$(x)$ during which the algorithm HALTS. There are two cases to consider here.

\paragraph{Case 1.} The activation decreases the weight $W_x$. Recall case 1 in the justification for Rule~\ref{rule:dirty:up}. This activation does not make the node $x$ dirty. It either remains in state $\Up$, or switches to state $\Upb$ if $1-2/\beta \leq W_x < 1-1/\beta$ after the activation. A call to UPDATE-STATUS$(x)$ following this activation will never make the algorithm HALT.

\paragraph{Case 2.} The activation increases the weight $W_x$. If it is a natural activation, then by Assumption~\ref{assume:weight} this increases the weight $W_x$ by at most $\beta^{-(\ell(x)+1)} \leq \beta^{-\ell(x)} - \beta^{-(\ell(x)+1)}$. See equation~\ref{eq:beta:K:L}.  Otherwise, the activation is induced, and suppose that it happens because some neighbour $y$ of $x$ decreases the value of $\ell_y(x, y)$ from $i+1$ to $i$ (say). Then this increases the weight $W_x$ by $\beta^{-i} - \beta^{-(i+1)} \leq \beta^{-\ell(x)} - \beta^{-(\ell(x)+1)}$. The inequality holds since we must have $i \geq \ell(x)$. We conclude that regardless of whether the activation of $x$ is natural or induced, it increases the weight $W_x$ by at most $\beta^{-\ell(x)} - \beta^{-(\ell(x)+1)}$. Immediately after this activation, the node $x$ becomes dirty (see Rule~\ref{rule:dirty:up}) and calls the subroutine FIX-DIRTY-NODE$(x)$, which in turn, calls the subroutine FIX-UP$(x)$. See Figures~\ref{new:fig:fix:dirty} and~\ref{new:fig:fix:up}. During the call to the subroutine FIX-UP$(x)$, the node $x$ up-marks some edge $(x, v) \in E_{\ell(x)}(x)$. This up-marking reduces the weight $W_x$ by exactly $\beta^{-\ell(x)} - \beta^{-(\ell(x) + 1)}$. Thus, this reduction in the value of $W_x$ is sufficient to compensate for the increase in the value of $W_x$ due to the preceding activation. Hence, when the subroutine UPDATE-STATUS$(x)$ eventually gets called after the up-marking of the edge $(x, y)$, we have $W_x < 1$. During the call to UPDATE-STATUS$(x)$, if we find out that $1-2/\beta \leq W_x < 1-1/\beta$, then we change the state of $x$ to $\Upb$. Else if we find out that $1-1/\beta \leq W_x < 1$ and $E_{\ell(x)}(x) = \emptyset$, then the node $x$ moves  to a higher level while remaining in state $\Up$. Finally, if we find out that $1-1/\beta \leq W_x < 1$ and $E_{\ell(x)}(x) \neq \emptyset$, then the node $x$ remains in the same level and in the same state $\Up$. Thus, under no situation does the call to UPDATE-STATUS$(x)$ makes our algorithm HALT. 

\subsection{Deriving a contradiction for $S = \Idle$.}

Just before time $t_1$, the node $x$ has $\state[x] = \Idle$, $M_{up}(x) = M_{down}(x) = \emptyset$, and $1-2/\beta \leq W_x < 1-1/\beta$. See row (4) of Table~\ref{fig:different:states}. Recall case 2 in the justification for Rule~\ref{rule:dirty:slack}. The node $x$ does not become dirty due to the activation at time $t_1$. It either remains in state $\Idle$, or switches to state $\Down$ if $f(\beta) \leq W_x < 1-1/\beta$, or switches to state $\Up$ if $1-1/\beta \leq W_x < 1$. A call to UPDATE-STATUS$(x)$ following this activation will never make the algorithm HALT.

\subsection{Deriving a contradiction for $S = \Slack$.}

Just before time $t_1$, the node $x$ has $\state[x] = \Slack$, $M_{up}(x) = M_{down}(x) = \emptyset$, $\ell(y) = K$, and $0 \leq W_x < f(\beta)$. See row (3) of Table~\ref{fig:different:states}. Recall case 1 in the justification for Rule~\ref{rule:dirty:slack}. The node $x$ does not become dirty due to the activation at time $t_1$. It either remains in state $\Slack$, or switches to state $\Down$ if $f(\beta) \leq W_x < 1-1/\beta$. A call to UPDATE-STATUS$(x)$ following this activation will never make the algorithm HALT.

\subsection{Deriving a contradiction for $S = \Downb$.}

The only way the node $x$ can change its level is if we execute the steps in Case 2-a or 2-b during a call to UPDATE-STATUS$(x)$. This situation can never occur during the time-interval $[t_0, t_1]$, throughout which we have $\state[x] = S = \Downb$. Thus, the node $x$ stays at the same level throughout the time-interval $[t_0, t_1]$. 

In Claims~\ref{new:cl:halt:downb:1} and~\ref{new:cl:halt:downb:2}, we respectively bound the weight $W_x$ at time-instants $t_0$ and $t_1$. In Corollary~\ref{cor:halt:downb}, we use these two claims to bound the change in the weight $W_x$ during the time-interval $[t_0, t_1]$.

\begin{claim}
\label{new:cl:halt:downb:1}
 $W_x < 1 - 2/\beta + 1/\beta^K$ at time $t_0$.
\end{claim}

\begin{proof}
There is an activation of the node $x$ at time $t_0$ which changes its state to $\Downb$. As per the discussion in Section~\ref{sub:sec:dirty:node}, the only way this can happen is if $\state[x] = \Down$ just before time $t_0$ and the activation at time $t_0$ increases  $W_x$ (see Case 1 in the justification for Rule~\ref{rule:dirty:down}). Thus, 
row (2) in Table~\ref{fig:different:states} gives: $W_x < 1 - 2/\beta$ just before time $t_0$.   As the weight of an edge is at most $\beta^{-K}$, the activation of   $x$ at time $t_0$ changes  $W_x$ by at most $\beta^{-K}$. Hence, $W_x < 1 - 2/\beta + 1/\beta^K$  after the activation of $x$ at time-instant $t_0$.
\end{proof}

\begin{claim}
\label{new:cl:halt:downb:2}
$W_x \geq 1 - 1/\beta$ and $M_{down}(x) \neq \emptyset$ at time $t_1$.
\end{claim}

\begin{proof}
$\state[x] = \Downb$ throughout the time-interval $[t_0, t_1]$. Just before time $t_1$, the node $x$ undergoes an activation, say, {\bf a*}. Subsequent to the activation {\bf a*}, our algorithm HALTS during a call to  UPDATE-STATUS$(x)$ at time $t_1$. From row (6) of Table~\ref{fig:different:states}, we get: $M_{down}(x) \neq \emptyset$, $M_{up}(x) = \emptyset$ and $1-2/\beta \leq W_x < 1-1/\beta$ just before the activation {\bf a*}. No edge gets inserted into the sets $M_{up}(x)$ and $M_{down}(x)$ during the activation {\bf a*}. Since the algorithm HALTS at time $t_1$, the activation {\bf a*} must have changed the weight $W_x$ in such a way that the node $x$ violates the constraints for every state as defined in Table~\ref{fig:different:states}. This can happen only if $W_x \geq 1 - 1/\beta$ and $M_{down}(x) \neq \emptyset$ at time $t_1$.
\end{proof}

\begin{corollary}
\label{cor:halt:downb}
During the interval $[t_0, t_1]$, the node-weight $W_x$  increases by at least $1/\beta - 1/\beta^K$.
\end{corollary}

\begin{proof}
Follows from Claims~\ref{new:cl:halt:downb:1} and~\ref{new:cl:halt:downb:2}.
\end{proof}

\begin{claim}
\label{new:cl:halt:downb:dirty}
An event of Type B does not make the node $x$ dirty. An event of Type A makes the node $x$ dirty.
\end{claim}

\begin{proof}
Throughout the time-interval $[t_0, t_1]$, we have $\state[x] = \Downb$. The claim follows from Rule~\ref{rule:dirty:up}.
\end{proof}

In the next three claims -- \ref{new:cl:halt:downb:activation:B},~\ref{new:cl:halt:downb:activation:C} and~\ref{new:cl:halt:downb:activation:C:failure} --  we bound the change in the weight $W_x$ that can result from an event of Type A or C. 

\begin{claim}
\label{new:cl:halt:downb:activation:B}
The node-weight $W_x$ increases by at most $\Delta = \beta^{-\ell(x)} - \beta^{-\ell(x)-1}$ due to a Type A event.
\end{claim}

\begin{proof}
If the Type A event occurs due a natural activation of $x$, then the claim follows from Assumption~\ref{assume:weight} since $\beta^{-\ell(x) - 1} \leq \beta^{-\ell(x)} - \beta^{-\ell(x)-1}$ as long as $\beta \geq 2$ (see equation~\ref{eq:beta:K:L}). For the rest of the proof, suppose that the Type A event occurs due to an induced activation. This means that the Type A event results from some neighbour $y$ of $x$ decreasing the value of $\ell_y(x, y)$ from, say, $(i+1)$ to $i$. Now, we consider the following  cases depending on the values of $i$ and $\ell_x(x, y)$.

\begin{itemize}
\item {\em Case 1.} $i \geq \ell(x)$. 

\smallskip
\noindent In this case,  it follows that the weight $W_x$ increases by $\beta^{-i} - \beta^{-(i+1)} \leq \Delta$ due to this Type A event. 

\item {\em Case 2.} $i = \ell(x) - 1$ and $\ell_x(x, y) \geq \ell(x)$. 

\smallskip
\noindent In this case, the weight $w(x, y)$ does not change as the value of $\ell_y(x, y)$ drops from $(i+1)$ to $i$. 

\item {\em Case 3.} $i = \ell(x) - 1$ and $\ell_x(x, y) = \ell(x) - 1$. 

\smallskip
\noindent Recall that the value of $\ell_y(x, y)$ drops from $(i+1)$ to $i$ due to a call to the subroutine PIVOT-DOWN$(y, (x, y))$, which is described in Figure~\ref{new:fig:pivot:down}. It is easy to check that in this case the call to PIVOT-DOWN$(y, (x, y))$ executes the steps (16) -- (20) in Figure~\ref{new:fig:pivot:down}, with $v = y$ and $u = x$. Hence, because of the undo operation performed by the node $x$ the weight $w(x, y)$ remains unchanged, and the call to PIVOT-DOWN$(y, (x,y))$  returns {\sc False} in step (20) in Figure~\ref{new:fig:pivot:down}.

\item {\em Case 4.} $i < \ell(x) - 1$.

\smallskip
\noindent Since $\ell_x(x, y) \geq \ell(x) - 1$, in this case the weight $w(x, y)$ does not change as the value of $\ell_y(x, y)$ decreases from $(i+1)$ to $i$.
\end{itemize}
\noindent Since $\ell_x(x, y) \geq \ell(x) - 1$, the above four cases cover all possible situations. It follows that only case 1 can happen during a Type A event, for in every other case the weight $W_x$ does not get changed. However, note that in case 1 the weight $W_x$ increases by at most $\Delta$. This concludes the proof of the lemma.
\end{proof}

Consider an event of Type C. This event occurs when we call the subroutine FIX-DIRTY-NODE$(x)$. Since $\state[x] = \Downb$, this in turn leads to a call to the subroutine FIX-DOWN-B$(x)$. See Figures~\ref{new:fig:fix:dirty} and~\ref{new:fig:fix:down:b}. Hence, during a Type C event, the node $x$ un-marks one or more incident edges $(x, y) \in M_{down}(x)$ by calling the subroutine PIVOT-UP$(x, (x, y))$. If the un-marking of an edge $(x, y)$ changes its weight $w(x, y)$, then we say that the un-marking is a {\em success}; otherwise the un-marking is a {\em failure}. Figure~\ref{new:fig:fix:down:b} ensures that an event of Type C leads to at most one success.

\begin{claim}
\label{new:cl:halt:downb:activation:C}
If a Type C event leads to a success, then it decreases the node-weight $W_x$ by $\beta^{-(\ell(x)-1)} - \beta^{-\ell(x)}$, which is more than $\Delta$ as defined in Claim~\ref{new:cl:halt:downb:activation:B}.
\end{claim}

\begin{proof}
Let the success correspond to the un-marking of the edge $(x, y)$.
 Just before this un-marking, we have $(x, y) \in M_{down}(x)$ and hence $\ell_x(x, y) = \ell(x) - 1$. The un-marking increases the value of $\ell_x(x, y)$ from $\ell(x) - 1$ to $\ell(x)$. For this to change the weight $w(x, y)$, the value of $\ell(x, y)$ must also have increased from $\ell(x)-1$ to $\ell(x)$ due to the un-marking. This means that the weight $w(x, y)$ decreases by an amount $\beta^{-(\ell(x)-1)} - \beta^{-\ell(x)}$ due to the un-marking. Since any Type C event leads to at most one success,  the weight $W_x$ also decreases by exactly this amount during this Type C event.
 \end{proof}

\begin{claim}
\label{new:cl:halt:downb:activation:C:failure}
If a Type C event does not lead to a success, then it does not change the weight $W_x$, and $\beta^5$ edges get deleted from the set $M_{down}(x)$ due to this Type C event. 
\end{claim}

\begin{proof}
Consider a Type C event  that does not lead to a success. During this event, each time the node $x$ un-marks an edge $(x, y)$, it leads to a failure and does not change the weight $w(x, y)$. Thus, the weight $W_x$ also does not change due to such an event of Type C.

Suppose that the Type C event under consideration  leads to zero success and less than $\beta^5$ failures. This  implies that the subroutine FIX-DOWN-B$(x)$ terminates due to step (05) in Figure~\ref{new:fig:fix:down:b}, and thus $M_{down}(x) = \emptyset$ at this point in time. Next,  the subroutine FIX-DIRTY-NODE$(x)$    calls UPDATE-STATUS$(x)$  as per step (09) in Figure~\ref{new:fig:fix:dirty}, which in turn changes the state of the node $x$ since we cannot simultaneously have $\state[x] = \Downb$ and $M_{down}(x) = \emptyset$. See row (6) of Table~\ref{fig:different:states}. However, this leads us to a contradiction, for we have assumed that $\state[x] = \Downb$ throughout the time-interval $[t_0, t_1]$.
\end{proof}

Let $n_A$ and $n_C$ respectively denote the number of Type A and Type C events during the time-interval $[t_0, t_1]$. Let $n_C^s$ (resp. $n_C^f$) denote the number of Type C events during the time-interval $[t_0, t_1]$ that lead  (resp. do not lead) to a success. Clearly, we have: $n_C = n_C^s + n_C^f$. 
By Claim~\ref{new:cl:halt:downb:dirty}, every Type A event is followed by a Type C event. Hence, we infer that $n_A \leq n_C = n_C^s + n_C^f$.  Rearranging the terms, we get:
\begin{equation}
\label{eq:F:2:downb}
n_C^f \geq n_A  - n_C^s 
\end{equation}
Any change in the weight $W_x$ during the time-interval $[t_0, t_1]$ results solely from an event of Type A, B or C.  Now,  an event of Type A  increases the weight $W_x$ by at most $\Delta$ (see Claim~\ref{new:cl:halt:downb:activation:B}), an event of Type B decreases the weight $W_x$,  an event of Type C that leads to a success decreases the weight $W_x$ by more than $\Delta$ (see Claim~\ref{new:cl:halt:downb:activation:C}), and an event of Type C that does not lead to a success leaves the value of $W_x$ unchanged (see Claim~\ref{new:cl:halt:downb:activation:C:failure}). Since the weight $W_x$ increases by at least $1/\beta - 1/\beta^K$ during the time-interval $[t_0, t_1]$ (see Corollary~\ref{cor:halt:downb}), we get:
\begin{equation}
\label{eq:F:3:downb}
(n_A - n_C^s) \cdot \Delta \geq 1/\beta - 1/\beta^K
\end{equation}
 Claim~\ref{new:cl:halt:downb:activation:B} gives: $1/\Delta \geq \beta^{\ell(x)}$. By eq~(\ref{eq:beta:K:L}), we have $1/\beta - 1/\beta^K \geq 1/\beta^2$. Thus, eq~(\ref{eq:F:2:downb}) and~(\ref{eq:F:3:downb}) give: \begin{equation}
\label{eq:F:4:downb}
n_C^f \geq (1/\Delta) \cdot (1/\beta - 1/\beta^K) \geq \beta^{\ell(x)-2} 
\end{equation}
By Claim~\ref{new:cl:halt:downb:activation:C:failure}, during each Type C event that contributes to $n_C^f$, the node $x$ deletes $\beta^5$ edges from the set $M_{down}(x)$. Hence, equation~\ref{eq:F:4:downb}  implies that during the time-interval $[t_0, t_1]$, the node $x$ deletes $n_C^f \cdot \beta^5 \geq \beta^{\ell(x)+3}$ edges from $M_{down}(x)$.  Furthermore, the node $x$ never inserts an edge into the set $M_{down}(x)$ during the time-interval $[t_0, t_1]$, for  $\state[x] = \Downb$ throughout the time-interval $[t_0, t_1]$. Thus, we have:
\begin{equation}
\label{eq:F:5:downb}
|M_{down}(x)| \geq \beta^{\ell(x)+3} \text{ at time-instant } t_0.
\end{equation} 
Note that every edge $(x, v) \in M_{down}(x)$ has $\ell_x(x, v) = \ell(x) - 1$ and $\ell_v(x, v) \leq \ell(x)$ by Invariant~\ref{main:inv:shadow:level}. Thus, the weight of every edge $(x, v) \in M_{down}(x)$ is given by $w(x, y) \geq \beta^{-\ell(x)}$. By equation~\ref{eq:F:5:downb}, we get: $W_x \geq \sum_{(x, v) \in M_{down}(x)} w(x,v) \geq |M_{up}(x)| \cdot \beta^{-\ell(x)} > 1$ at time-instant $t_0$. This leads to a contradiction, since $\state[x] = \Downb$ at time $t_0$ and hence row(6) of Table~\ref{fig:different:states} requires that $W_x < 1-1/\beta$.

\subsection{Deriving a contradiction for $S = \Down$.}
\label{new:sec:last}

Suppose that when the node $x$ enters state $\Down$ at time $t_0$, we have $\ell(x) = i > K$. This means that the node was either in state $\Downb$ or in state $\Idle$ just before time $t_0$ (see the discussion in Section~\ref{sub:sec:dirty:node}). Hence, the weight $W_x$ is very close to  $1-2/\beta$ at time $t_0$.

If an adversary wants to ensure that our algorithm HALTS at time $t_1$, then her best bet is to apply the following strategy. Keep reducing the weight $W_x$ in a series of activations, and hope that by time $t_1$ we reach a scenario where: $W_x < f(\beta) = 1-3/\beta$, $\ell(x) = j > K$ and there are many edges $(x, y)$ with $\ell_y(x, y) < j$. This prevents the node $x$ from going down to level $K$ and switch to state $\Slack$. Hence, the algorithm must HALT at this point.

For the above strategy to work out as per plan, the adversary must be able to reduces the weight $W_x$ from being close to $1-2/\beta$ at time $t_0$ to being less than $f(\beta) = 1-3/\beta$ at time $t_1$. In other words, by a series of activations the adversary must be able to reduce the weight $W_x$ by $1/\beta$ during the time-interval $[t_0, t_1]$.

Each of these activations reduce the weight $W_x$ by at most $\beta^{-(\ell(x)-1)} - \beta^{-\ell(x)}$: this bound is achieved when an edge $(x, y) \in M_{down}(x)$ has the value of $\ell_y(x, y)$ increased from $\ell(x) - 1$ to $\ell(x)$.\footnote{This always results from an induced activation. By Assumption~\ref{assume:weight}, a natural activation reduces the weight $W_x$ by at most $\beta^{-(\ell(x)+1)} \leq \beta^{-(\ell(x) - 1)} - \beta^{-\ell(x)}$. } After every such  activation, we call the subroutine FIX-DIRTY-NODE$(x)$, which in turn calls FIX-DOWN$(x)$. The call to FIX-DOWN$(x)$ consists of a series of down-markings of edges $(x, y) \in E_{\ell(x)}(x) \setminus M_{down}(x)$ by the node $x$. If a down-marking results in the weight $W_x$ getting increased, then we say that the down-marking is a {\em success}, else we say that the down-marking is a {\em failure}. Each call to FIX-DOWN$(x)$ results in either one success or $\beta^5 L$ failures. In case of a success, the weight $W_x$ increases by $\beta^{-(\ell(x)-1)} - \beta^{\ell(x)}$, for the level of the edge being down-marked by $x$ changes from $\ell(x)$ to $\ell(x)-1$. Note that in case of a success, the amount by which $W_x$ increases is sufficient to compensate for the decrease in $W_x$ that lead to the call to FIX-DOWN$(x)$. To summarise, the adversary ensures that each activation of $x$ reduces the weight $W_x$ by at most $\beta^{-(\ell(x)-1)} - \beta^{-\ell(x)}$. This is followed by a call to FIX-DOWN$(x)$, which either brings the weight $W_x$ back to its initial value, or results in $\beta^5 L$ failures. In other words, each time the weight $W_x$ decreases by $\beta^{-(\ell(x) - 1)} - \beta^{-\ell(x)} \leq \beta^{-(\ell(x) -1)}$, the node $x$ down-marks $\beta^5 L $ many new edges from $E_{\ell(x)}(x) \setminus M_{down}(x)$. Initially, the node $x$ cannot have more than $\beta^{-\ell(x)}$ many edges in $E_{\ell(x)}(x)$, for each such edge has weight $\beta^{-\ell(x)}$ and the node-weight $W_x$ clearly cannot exceed one. Thus, after $\beta^{\ell(x)}/(\beta^5 L)$ many such events, each of which reduces the weight $W_x$ by at most $\beta^{-(\ell(x) - 1)}$, we would have $E_{\ell(x)}(x) \setminus M_{down}(x) = \emptyset$. At that point the node $x$ will move down one level below. During this interval, the weight $W_x$ would drop by at most $(\beta^{\ell(x)}/(\beta^5 L) ) \cdot \beta^{-(\ell(x) - 1)} \leq 1/(\beta^4 L)$.

Thus, we infer that each time the weight $W_x$ drops by $1/(\beta^4 L)$, the level of $x$ drops by one. Accordingly, much before the weight $W_x$ drops from $1-2/\beta$ to below $f(\beta) = 1 - 3/\beta$, the node $x$ would reach level $K$. This leads to a contradiction.

\section{Proof of Theorem~\ref{th:update:time}}
\label{new:sec:th:update:time}

We first introduce the concept of the {\em down-level} $\ell^*(x)$ of a node $x \in V$. This is defined below.
\begin{equation}
\label{eq:down-level}
\ell^*(x)  = \begin{cases} \ell(x) - 1 & \text{ if } M_{down}(x) \neq \emptyset; \\
\ell(x) & \text{ otherwise}.
\end{cases}
\end{equation}
Recall that if $\ell(x) = K$, then $M_{down}(x) = \emptyset$.\footnote{This holds since every edge $(x, y) \in M_{down}(x)$ has $\ell_x(x,y) = \ell(x) - 1$, and since $\ell_x(x, y) \in [K, L]$.} Hence, we get: $\ell^*(x) \in [K, L]$ for every node $x \in V$.  We will prove the following lemma.

\begin{lemma}
\label{lm:update:time}
After two consecutive iterations of the {\sc While} loop in Figure~\ref{new:fig:dirty}, the value of $\ell^*(x)$ decreases by at least one.
\end{lemma}

Since there are $(L-K+1)$ possible values for $\ell^*(x)$,  Lemma~\ref{lm:update:time} implies that the {\sc While} loop in Figure~\ref{new:fig:dirty} runs for at most $2(L-K+1) = O(\log n)$ iterations. By Lemma~\ref{lm:fix:dirty:node}, each iteration of the {\sc While} loop in Figure~\ref{new:fig:dirty} takes $O(\log^2 n)$. This gives a total worst-case update time of $O(\log^3 n)$ for handling one edge insertion or deletion in the input graph. 
We devote the rest of this section to the proof of Lemma~\ref{lm:update:time}. 

Consider three consecutive iterations of the {\sc While} loop in Figure~\ref{new:fig:dirty}. For $i \in \{1, 2, 3\}$, let $x_i$ be the node considered in the $i^{th}$ iteration. Furthermore, for $i \in \{1, 2, 3\}$, let $\ell^*_i$ be the value of $\ell^*(x_i)$ just before the call to the subroutine FIX-DIRTY-NODE$(x_i)$ in the $i^{th}$ iteration. We will show that $\ell^*_1 > \ell^*_3$. As per Figure~\ref{new:fig:fix:dirty}, we must have $\state[x_1] \in \{ \Up, \Downb, \Upb, \Down \}$ just before the $1^{st}$ iteration of the {\sc While} loop in Figure~\ref{new:fig:dirty}.  In order to prove Lemma~\ref{lm:update:time}, we  consider two mutually exclusive and exhaustive  cases.

\paragraph{Case 1.} $\state[x_1] \in \{\Up, \Downb\}$ just before the $1^{st}$ iteration of the {\sc While} loop in Figure~\ref{new:fig:dirty}. 

\medskip
\noindent In this case, Lemma~\ref{cor:new:lm:pivot:up} implies that $\ell^*_1 \geq \ell^*_2$ and $\state[x_2] \in \{ \Down, \Upb \}$ just before the $2^{nd}$ iteration of the {\sc While} loop in Figure~\ref{new:fig:dirty}. Now, applying Lemma~\ref{cor:new:lm:pivot:down}, we get: $\ell^*_1 \geq \ell^*_2 > \ell^*_3$. 

\paragraph{Case 2.} $\state[x_1] \in \{\Down, \Upb\}$ just before the $1^{st}$ iteration of the {\sc While} loop in Figure~\ref{new:fig:dirty}. 

\medskip
\noindent In this case, Lemma~\ref{cor:new:lm:pivot:down} implies that $\ell^*_1 > \ell^*_2$  just before the $2^{nd}$ iteration of the {\sc While} loop in Figure~\ref{new:fig:dirty}. Now, applying Corollary~\ref{cl:final}, we again get: $\ell^*_1 > \ell^*_2 \geq \ell^*_3$.

\begin{lemma}
\label{cor:new:lm:pivot:up}
Consider a call to  FIX-DIRTY-NODE$(v)$, and suppose that  $\state[v] \in \{ \Up, \Downb \}$ and $\ell^*(v) = i$ just before the call. If a neighbour $u$ of $v$ becomes dirty because of this call, then  $\ell^*(u) \leq i$ and $\state[u] \in \{ \Upb, \Down\}$ at the termination of the subroutine FIX-DIRTY-NODE$(v)$.
\end{lemma}

The proof of Lemma~\ref{cor:new:lm:pivot:up} appears in Section~\ref{sec:up}.

\begin{lemma}
\label{cor:new:lm:pivot:down}
Consider a call to  FIX-DIRTY-NODE$(v)$, and suppose that  $\state[v] \in \{ \Down, \Upb \}$ and $\ell^*(v) = i$ just before the call. If a neighbour $u$ of $v$ becomes dirty because of this call, then  $\ell^*(u) < i$  at the termination of the subroutine FIX-DIRTY-NODE$(v)$.
\end{lemma}

The proof of Lemma~\ref{cor:new:lm:pivot:down} appears in Section~\ref{sec:down}.

\begin{corollary}
\label{cl:final}
Consider a call to  FIX-DIRTY-NODE$(v)$, and suppose that   $\ell^*(v) = i$ just before the call. If a neighbour $u$ of $v$ becomes dirty because of this call, then  $\ell^*(u) \leq i$  at the termination of the subroutine FIX-DIRTY-NODE$(v)$.
\end{corollary}

\begin{proof}
Since we have called  FIX-DIRTY-NODE$(v)$, we must have $\state[v] \in \{ \Up, \Downb, \Up, \Down \}$ as per Figure~\ref{new:fig:fix:dirty}. The corollary now follows from Lemmas~\ref{cor:new:lm:pivot:up} and~\ref{cor:new:lm:pivot:down}.
\end{proof}

\subsection{Proof of Lemma~\ref{cor:new:lm:pivot:up}.}
\label{sec:up}

See Figures~\ref{new:fig:fix:up} and~\ref{new:fig:fix:down:b}. Since $\state[v] \in \{ \Up, \Downb \}$, a neighbour $u$ of $v$ can become dirty only if we call PIVOT-UP$(v, (u,v))$ during the execution of FIX-DIRTY-NODE$(v)$. Recall that $\ell^*(v) = i$, and consider two possible cases.

\paragraph{Case 1.} $\state[v] = \Up$. In this case, we have $M_{down}(v) = \emptyset$ as per row (1) of Table~\ref{fig:different:states}. So equation~\ref{eq:down-level} implies that $\ell(v) = i$. In the call to PIVOT-UP$(v, (u,v))$, the node $v$ up-marks the edge $(u,v)$  by increasing the value of $\ell_v(u,v)$ from $i$ to $(i+1)$. By Claim~\ref{new:lm:pivot:up}, the node $u$ becomes dirty only if $\ell^*(u) \leq i$ and $\state[u] \in \{ \Down, \Upb\}$. 

\paragraph{Case 2.} $\state[v] = \Downb$. In this case, we have $M_{down}(v) \neq \emptyset$ as per row (6) of Table~\ref{fig:different:states}. So equation~\ref{eq:down-level} implies that $\ell(u) = i+1$. In the call to PIVOT-UP$(v, (u,v))$, the node $v$ un-marks the edge $(u,v)$ by increasing the value of  $\ell_v(u,v)$ from $i$ to $(i+1)$. By Claim~\ref{new:lm:pivot:up}, the node $u$ becomes dirty only if $\ell^*(u) \leq i$ and $\state[u] \in \{ \Down, \Upb\}$. 

\medskip
\noindent 
 Lemma~\ref{cor:new:lm:pivot:up} follows as we terminate the subroutine FIX-DIRTY-NODE$(v)$ just after $u$ becomes dirty.

\begin{claim}
\label{new:lm:pivot:up}
Consider a call to the subroutine PIVOT-UP$(v, (u,v))$ as  in Figure~\ref{new:fig:pivot:up}. Suppose that  $\ell_v(u,v)$ increase from $j$ to $(j+1)$ during step (01). If  $u$ becomes dirty during  this call, then  $\ell^*(u) \leq j$ and $\state[u] \in \{ \Upb, \Down \}$ at the end of the subroutine.
\end{claim}

\begin{proof}
Suppose that the node $u$ becomes dirty during the execution of the subroutine PIVOT-UP$(v, (u,v))$. This can happen only due to the execution of step (03)  in Figure~\ref{new:fig:pivot:up}. Thus, step (02) in Figure~\ref{new:fig:pivot:up} ensures that $Y = \text{{\sc True}}$ and \{either $\state[u] = \Upb$ or ($\state[u] =  \Down$ and $\ell(u) > K$)\}.  It follows that the weight $w(u,v)$ and the level $\ell(u,v)$ change during the execution of step (01), which increases  $\ell_v(u,v)$ from $j$ to $(j+1)$. Hence, we must have $\ell_u(u,v) \leq j$ just before step (01). Since $\ell(u) \leq \ell_u(u,v) + 1$, we also infer that $\ell(u) \leq j+1$ just before step (01). We consider two cases, depending on the values of $\ell(u)$ and $\ell_u(u,v)$. 

\medskip
\noindent
{\em Case 1.} $\ell(u) \leq j$ just before step (01).  The value of $\ell(u)$ does not change during the execution of PIVOT-UP$(v, (u,v))$. Hence, even at the end of the subroutine, we have $\ell^*(u) \leq \ell(u) \leq j$. 

\medskip
\noindent {\em Case 2.} $\ell(u) = j+1$ and $\ell_u(u,v) = j$ just before step (01). In this case, the increase in the value of $\ell_v(u,v)$  from $j$ to $(j+1)$ during step (01) does not lead to any change in the value of $\ell_u(u,v)$. In other words, since $\ell(u)$ remains larger than or equal to $\ell_v(u,v)$ even after the increase in the value of $\ell_v(u,v)$, steps (08) -- (09) in Figure~\ref{new:fig:move:up} do not get executed. Hence, even after step (01) in Figure~\ref{new:fig:pivot:up}, we continue to have $\ell(u) = j+1$ and $\ell_u(u, v) = j$. In fact, we continue to have $\ell(u) = j+1$ and $\ell_u(u,v) = j$ till the end of the subroutine. At that point in time, we conclude that  $M_{down}(u) \neq \emptyset$ since $(u,v) \in M_{down}(u)$, and hence $\ell^*(u) = \ell(u) - 1 = j$. 
\end{proof}

\subsection{Proof of Lemma~\ref{cor:new:lm:pivot:down}.}
\label{sec:down}

See Figures~\ref{new:fig:fix:down} and~\ref{new:fig:fix:up:b}. Since $\state[v] \in \{ \Down, \Upb \}$, a neighbour $u$ of $v$ can become dirty only if we call PIVOT-DOWN$(v, (u,v))$ during the execution of the subroutine FIX-DIRTY-NODE$(v)$. Recall that $\ell^*(v) = i$. We now consider two possible cases.

 \paragraph{Case 1.}  $\state[v] = \Down$. In this case, equation~\ref{eq:down-level} implies that $\ell(v) \in \{i, i+1\}$. In the call to PIVOT-DOWN$(v, (u,v))$, the node $v$  down-marks the edge $(u,v)$. Let $j$ be the new value of $\ell_v(u,v)$ after this down-marking, i.e., the value of $\ell_v(u,v)$ drops from $(j+1)$ to $j$. Then we clearly have: $j \leq i$. By Claim~\ref{new:lm:pivot:down},  $u$ becomes dirty only if $\ell^*(u) < j \leq i$.

 \paragraph{Case 2.} $\state[v] = \Upb$. In this case, we have $M_{down}(v) = \emptyset$ as per row (5) of Table~\ref{fig:different:states}. Hence, equation~\ref{eq:down-level} implies that $\ell(v) = i$. In the call to PIVOT-DOWN$(v, (u,v))$,  the node $v$ un-marks the edge $(u, v)$ by decreasing the value of  $\ell_v(u,v)$ from $(i+1)$ to $i$. By Claim~\ref{new:lm:pivot:down},  $u$ becomes dirty only if $\ell^*(u) < i$. 
 
 \medskip
 \noindent
Lemma~\ref{cor:new:lm:pivot:down} follows as we terminate the subroutine FIX-DIRTY-NODE$(v)$ just after $u$ becomes dirty.

\begin{claim}
\label{new:lm:pivot:down}
Consider a call to the subroutine PIVOT-DOWN$(v, (u,v))$ as in Figure~\ref{new:fig:pivot:down}, and let step (01) decrease  $\ell_v(u,v)$  from $(j+1)$ to $j$. If  $u$ becomes dirty due to this call, then  $\ell^*(u) < j$ at the end of the subroutine. 
\end{claim}

\begin{proof}
Suppose that the node $u$ becomes dirty during the call to the subroutine PIVOT-DOWN$(v, (u,v))$.  Depending on the state of the node $u$ in the beginning of Figure~\ref{new:fig:pivot:down}, we consider three possible cases.

\medskip
\noindent {\em Case 1.} $\state[u] = \Up$. 

\noindent Here, the node $u$ can become dirty only due to the execution of step (09) in Figure~\ref{new:fig:pivot:down}. For the rest of this paragraph, we assume that step (09) gets executed. Thus, in step (03) we get $Y = \text{{\sc True}}$, and this ensures that the call to  MOVE-DOWN($v, (u,v)$) in step (01) changes the weight of the edge $(u,v)$. Recall that during step (01) the value of $\ell_v(u,v)$ changes from $(j+1)$ to $j$. For this event to change the weight  $w(u,v)$, we get: $\ell_u(u,v) \leq j$ just before  step (01). Since $\state[u] = \Up$, we must have $\ell(u) \leq \ell_u(u,v) \leq j$ at the same point in time. Now, there are two possibilities. Either $\ell(u) = j$ or $\ell(u) < j$ just before step (01).
\begin{enumerate}
\item Suppose that $\ell(u) = j$ just before step (01).  Since $\ell(u) \leq \ell_u(u,v) \leq j$ at the same point in time, we infer that $\ell(u) = \ell_u(u,v) = j = \ell_v(u,v)$ just after step (01). This, in turn, implies that  $(u,v) \notin M_{up}(u)$ and $\ell(u) \geq \ell_v(u,v)$ just after step (01). It follows that  due to step (04) we would never execute step (09). This leads to a contradiction. Hence, it cannot be the case that $\ell(u) = j$.
\item The only remaining possibility is that $\ell(u) < j$. In this case,  we have $\ell^*(u) \leq \ell(u) < j$. 
\end{enumerate}
To summarise,  if $\state[u] = \Up$, then  $u$  becomes dirty only if $\ell^*(u) < j$ at the end of the subroutine.

\medskip
\noindent {\em Case 2.} $\state[u] = \Downb$. 

\noindent Here,  the node $u$ can become dirty only due to the execution of step (22) in Figure~\ref{new:fig:pivot:down}. For the rest of this paragraph, we assume that step (22) gets executed. Thus, in step (16) we get $Y = \text{{\sc True}}$, and this ensures that the call to  MOVE-DOWN($v, (u,v)$) in step (01) changes the weight of the edge $(u,v)$. Recall that step (01) changes the value of $\ell_v(u,v)$  from $(j+1)$ to $j$. For this event to change the weight of the edge $(u,v)$, we must have $\ell_u(u,v) \leq j$ just before step (01). Since  $\ell(u) \leq \ell_u(u,v)+1$, we infer that $\ell(u) \leq j+1$ and $\ell_u(u,v) \leq j$ just before step (01). Accordingly, we consider the following possibilities.
\begin{enumerate}
\item $\ell(u) = j+1$ and $\ell_u(u,v) = j$. In this case, we have $(u,v) \in M_{down}(u)$ and $\ell_v(u,v) = j < \ell(u)$ during step (17). Hence, step (22) never gets executed, and we reach a contradiction. 
\item $\ell(u) = j$. In this case, since $\state[u] = \Downb$, we must have $\ell^*(u) = \ell(u) - 1 < j$. 
\item $\ell(u) < j$. In this case, we have $\ell^*(u) \leq \ell(u) < j$. 
\end{enumerate}
To summarise, if $\state[u] = \Downb$, then  $u$  becomes dirty only if $\ell^*(u) < j$ at the end of the subroutine.

\medskip
\noindent {\em Case 3.} $\state[u] \in \{\Upb, \Down, \Idle, \text{{\sc Slack}}\}$. 

\noindent Here, steps (28)-(30)  ensure that  $u$ never becomes dirty.
\end{proof}

\bibliographystyle{abbrv}
\bibliography{citations}

\end{document}